\documentclass[twocolumn]{autart}    

\usepackage{graphicx}          
\usepackage{etoolbox}
\usepackage{textcomp}
\usepackage{cite}
\usepackage{indentfirst}
\usepackage{amsmath,amssymb,amsfonts}
\usepackage{mathrsfs}
\usepackage{subfigure}
\usepackage{bm}
\usepackage[noend]{algpseudocode}
\usepackage{algorithmicx}
\usepackage{algorithm}
\usepackage[colorlinks,linkcolor=black,anchorcolor=black,citecolor=black]{hyperref}
\usepackage{algpseudocode}
\algdef{SE}[DOWHILE]{Do}{doWhile}{\algorithmicdo}[1]{\algorithmicwhile\ #1}%

\usepackage{booktabs}
\usepackage{ntheorem}
\newtheorem{theorem}{Theorem}
\newtheorem{lemma}{Lemma}
\newtheorem{assumption}{Assumption}
\newtheorem{proof}{Proof}
\newtheorem{definition}{Definition} 
\newtheorem{remark}{Remark}
\newtheorem{problem}{Identification Problem}

\begin{document}

\begin{frontmatter}

\title{A Two-stage Identification Method for Switched Linear Systems\thanksref{footnoteinfo}}

\thanks[footnoteinfo]{This paper was not presented at any IFAC 
meeting. 
}


\author[a]{Zheng Wenju}\ead{zhengwj18@mails.tsinghua.edu.cn},
\author[a]{Ye Hao}\ead{haoye@tsinghua.edu.cn}
\address[a]{Department of Automation, Tsinghua University, Beijing}

\begin{keyword}
    Switched linear system, System Identification, Sparsity Optimization, Persistent Excitation Condition
\end{keyword}

\begin{abstract}                          
    In this work, a new two-stage identification method based on dynamic programming and sparsity inducing is proposed for switched linear systems. 
    Our method achieves sparsity inducing in the identification of switched linear systems by the constrained switching mechanism, in contrast to previous optimization-based identification techniques that rely on the rigid data distribution assumption in the parameter space.
    The proposed mechanism assumes the existence of a minimal interval between adjacent switching instants.
    First, an efficient iterative dynamic programming approach is used to determine the switching instants and segments using the constrained switching mechanism. Then, each submodel is identified as a combinatorial $\ell_0$ optimization problem, and the true parameter for each submodel is determined by solving the problem. 
    The problem of combinatorial $\ell_0$ optimization is solved by relaxing it into a convex $\ell_1$-norm optimization problem.
    Furthermore, the unbiasedness of the switched linear system identification is discussed thoroughly with the constrained switching mechanism and a new persistent excitation condition is proposed. Simulation experiments are conducted to indicate that our algorithms exhibit strong robustness against noise.
\end{abstract}

\end{frontmatter}

\section{Introduction}
Identification of switched linear systems has attracted significant interest in recent years and has been widely applied to real-world problems \cite{massucci2022statistical}, such as computer vision, benchmark detection, system biology, electromechanical systems, etc. This interest stems from the fact that the complex inherently nonlinear systems and phenomenon can be effectively modeled as a combination of a finite number of linear systems \cite{garulli2012survey}. Generally speaking, a switched linear system which comprises a finite number of linear subsystems and a switched mechanism determining the active subsystem at each time instant, can be parameterized and formulated as \cite{bako2011identification}:
\begin{subequations} \label{s1}
    \begin{align}
        y_k & = \theta_{\lambda_t}^\top x_k+e_k\\
        x_k & = \left[y_{k-1},\cdots,y_{k-n_a},u\top_{k-1},\cdots,u\top_{k-n_b}\right]\top
    \end{align}
\end{subequations}
where $x_k\in \mathbb{R}^n$ is the regressor vector, $u_k \in \mathbb{R}$ and $y_k \in \mathbb{R}$ are designated as the input and the output of the system, $n_a$ and $n_b$ are the orders of the systems and $e_k \in \mathbb{R}$ is an unknown noise. $\lambda_k\in\{1,\cdots,S\}$ is the discrete state which indicates the active subsystem at time $k$ where $S$ is the number of the subsystems and $\theta_{\lambda_k}\in \mathbb{R}^n,n=n_a+n_b$ is the parameter vector corresponding to the active subsystem at time $k$. 

In this paper, we focus on the identification problem of switched linear system from a collection of time series input-output data. 
Besides the estimation of each subsystem, the identification of the switched linear system also need to confirm the discrete state which is generated by the switched mechanism and indicates the active subsystem at each time. Therefore, the identification of the switched linear system is essentially a mixed integer programming problem. Moreover, the identification of the number of the switching segments and submodels remains a critical and open issue \cite{massucci2022statistical}. 

\subsection{Prior works on the identification of switched linear systems}
There have been a good deal of literature on the identification of switched linear systems. 
In recent years, data-driven methods are widely used to the identification of switched linear systems, such as model selection techniques based on statistical learning \cite{massucci2020structural,massucci2022statistical}, regression trees \cite{de2021stability}, self-adaptative multi-kernel algorithm \cite{sellami2019self}, semi-supervised learning approach \cite{du2020semi} etc. However, these methods lack of convergence analysis. 
The algebraic recursive methods are proposed in literature \cite{vidal2003algebraic,vidal2004recursive,bako2008algebraic, goudjil2017recursive,bako2011recursive, ozay2015set} by transforming a switched linear system into a lifted linear model independent of the switching sequence. 
The optimization based methods address the identification problem of switched linear systems through formulations involving sum-of-norms regularization \cite{ohlsson2013identification}, mixed-integer programming \cite{roll2004identification}, sparse optimization \cite{bako2011identification}, and so forth. 
Sparse optimization has been widely used in inverse problems and signal and image processing problems \cite{bruckstein2009sparse} in recent years and the identification algorithm for switched linear systems based on sparse optimization is proposed in \cite{bako2011identification}, leveraging a specific combinatorial $\ell_0$-norm optimization problem, where the optimal solution is related to the subsystem parameters. This decouples the identification of system parameters and the determination of the discrete states in the switched linear system.
The unbiasedness of the identification in the algebraic and optimization-based methods mentioned above relies on the strict assumption of data distribution in the parameter space, as follows. 
\begin{assumption} \label{assumption_ori}
    (Assumption 6 in \cite{bako2011identification}) There is no data pair $(x_k,y_k)$ that fits two different submodels of the switched linear system \eqref{s1}, i.e., $y_k = \theta_i^\top x_k+e_k = \theta_j^\top x_k+e_k \Rightarrow i=j.$
\end{assumption}
Assumption \ref{assumption_ori} requires that all regressor vectors $\{x_k\}$ cannot lie on the intersection of the planes formed by the parameters of each subsystem. Based on Assumption \ref{assumption_ori}, the persistent excitation condition for the identification of the switched linear systems are proposed in \cite{mu2022persistence} in the absence of noise. 
However, Assumption \ref{assumption_ori} cannot be guaranteed in practical applications, especially in the presence of noise.

\subsection{Contributions of this paper}
In this paper, we address the identification of switched linear systems based on a constrained switching mechanism, which assumes the existence of a minimal interval between adjacent switching instants. Considering the identification of the switching instants, we propose a two-stage identification framework for switched linear systems, departing from the assumption of the distribution of data in parameter space, as described in Assumption \ref{assumption_ori}.

Firstly, inspired by the qualitative trend analysis \cite{zhou2017new} and curve fitting \cite{bellman1969curve,bellman1961approximation,shawe2004kernel}, an efficient iterative dynamic programming algorithm proposed in our previous work \cite{zheng2023identification} is utilized to identify the switching instants under the constrained switching mechanism. Meanwhile, the collected time series data can be partitioned into different time segments according to the estimated switching instants. Subsequently, in the second stage inspired by the sparse optimization identification algorithm \cite{bako2011identification}, the parameters of each submodel can be identified by solving a specific combinatorial $\ell_0$-norm optimization problem based on the segmented data. The combinatorial $\ell_0$-norm optimization problem is solved by relaxing it into a convex $\ell_1$-norm optimization problem.
Furthermore, the unbiasedness of the proposed two-stage algorithm is guaranteed under the assumption of the constrained switching mechanism and a new sufficient condition based on the constrained switched mechanism for the persistent excitation condition of switched linear systems is proposed as a comparison of that in review \cite{mu2022persistence} based on Assumption \ref{assumption_ori}.

\subsection{Outline of this paper}
This paper is organized as follows. In the next section, the identification problem of the switched linear system is presented, and the conceptions of the sparse optimization are briefly discussed. Later in Section III, a two-stage identification framework is proposed based on both dynamic programming and sparse optimization. The unbiasedness of the identification algorithms based on $\ell_0$-norm optimization and $\ell_1$-norm optimization are detailed discussed. In Section IV, a new persistent excitation condition for switched linear systems is proposed. Experimental results based on simulation and real data of high speed trains are presented in Section V. Finally, brief conclusions are drawn in Section VI.

\section{Preliminaries and problem formulation} \label{section2}
In this section, we introduce the basic conceptions and notations of the sparsity presentation, which will be extensively used throughout the rest of the paper. The identification problem of the switched linear system will be also formulated.
\subsection{Sparse Optimization}
Recent years, sparse optimization has been widely used in inverse problems and signal and image processing problems \cite{bruckstein2009sparse}. 
The general optimization problem with the equations of linear systems can be formulated as 
\begin{equation} \label{P_J}
    (P_J): \min_z J(z) ~ \text{subject to } Az=b
\end{equation}
where $A$ is always given as a full-rank matrix and the objective function $J(\cdot)$ is selected to guarantee a unique solution. In this section, we primarily introduce the most commonly used sparse optimization criterion, $\ell_0$ norm and $\ell_1$ norm.
The problem with $J(z)=\lVert z \rVert_0$, representing the nonzero entries of the vector $z$, is denoted as the $\ell_0$-norm optimization problem, $P_0$, as well as the $\ell_1$-norm optimization problem, $P_1$, with the objective function $J(z)=\lVert z \rVert_1$, representing the sum of the absolute values of the elements in the vector $z$.

There are some measures of how sparse the columns of the given matrix are and here we mainly introduce the $spark$, mutual coherence and $k$-genericity index \cite{bako2011identification,bruckstein2009sparse} for the convenience of the uniqueness of sparse solutions as well as the unbiasedness of the identification which will be discussed in the later Section III. The definition of $spark$ is written as: 
\begin{definition}
    The spark of a given matrix $A$ is the smallest number of columns from $A$ that are linearly dependent, written as $spark(A)$.
\end{definition}
In fact, the $spark$ of the given matrix $A$ can be formulated as \cite{bako2011identification}
\begin{equation}
    spark(A)=\min_{z\in \ker(A),z\neq0} \lVert z \rVert_0
\end{equation}
where $\ker(A)$ represents the kernel space or the null space of the matrix $A$ and $\lVert z \rVert_0$ represents the $\ell_0$-norm of the vector $z$ which is inherently the number of the nonzero entries of the vector $z$. There is a crucial theorem about $spark$ as follows:
\begin{theorem} \label{theorem_spark}
    (See Theorem 2 in \cite{bruckstein2009sparse}) If a system of linear equations $Az=b$ with full row rank matrix $A$ has a solution $z$ obeying $\lVert z \rVert_0 <  {spark(A)}/ {2}$, this solution is necessarily the sparsest possible. 
\end{theorem}
The $spark$ gives a simple criterion for uniqueness of sparse solutions and this criterion forms the basis of many results in the field of the sparse optimization. A simpler way to guarantee the uniqueness is to exploit the mutual coherence of the given matrix $A$, defined as follows.
\begin{definition}
    The mutual coherence of a given matrix $A$ is the largest absolute normalized inner product between different columns from $A$. Denoting the $k$th column in $A$ by $a_k$, the mutual coherence can be formulated as 
    \begin{equation} \label{mu}
        \mu(A) = \max_{{1\leq i,j \leq m},~{i\neq j}} \frac{|a_i^\top a_j|}{\|a_i\|_2\|a_j\|_2}.
    \end{equation}
\end{definition}
By introducing the definition of mutual coherence, an analogue of Theorem \ref{theorem_spark} can be formulated as follows:
\begin{theorem} \label{theorem_mu}
    (See Theorem 5 and Theorem 7 in \cite{bruckstein2009sparse}) If $Az=b$, where $A$ is a matrix with full row rank, and there exists a solution $z$ obeying 
    \begin{equation} \label{condition_mu}
        \|z\|_0< \frac{1}{2} (1+\frac{1}{\mu(A)})
    \end{equation}
    then this solution is necessarily the sparsest possible. Moreover, it represents the unique solution to the sparse problem \eqref{P_J} where $J(z) = \lVert z\rVert_0$ as well as when $J(z) = \lVert z\rVert_1$.
\end{theorem}
The $k$-genericity index of the given matrix $A$, denoted as $v_k(A)$, is another way to measure the linear independence of the columns of $A$ and is first introduced in \cite{bako2011identification}.
\begin{definition} \label{ge_matrix}
    (See Definition 1 in \cite{bako2011identification}) For a given matrix $A\in \mathbb{R}^{n\times N},~n\leq N$ with the $k$th column denoted as $a_k$, the $k$-genericity index, $v_k(A)$ is the minimum integer $m$ such that any $n \times m$ submatrix of $A$ has rank $k$:
    \begin{equation} \label{k_index}
        \begin{aligned}
            v_k(A) = &\min \left\{  m: \forall (t_1,\cdots,t_m) \text{ with } t_i\neq t_j \text{ for } i\neq j,  \right.\\ 
            &  rank[a_{t_1},\cdots,a_{t_m}]=k \left. \right\}.
        \end{aligned}
    \end{equation}
    Moreover, if $k>rank(A)$, $v_k(A)=+\infty$ and if $k=0$, $v_0(A)=0$.
\end{definition}


\subsection{Identification Problem}
As mentioned in Section I, the identification of switched linear system \eqref{s1} involves not only identifying the parameters of each subsystem, $\{\theta_i\}_{i=1}^S$, but also determining the structure of the systems as well as the number of the subsystems, $S$, and the switching mechanism, which refers to the integer discrete states, $\{\lambda_k\}_{k=1}^N$. 
\begin{remark}
  The orders of the switched linear system \eqref{s1}, denoted as $n_a$ and $n_b$, are assumed to be known and equal for all subsystems which is adopted in most methods \cite{de2021stability,ozay2011sparsification,du2020semi,bako2011identification,bako2008algebraic}.
\end{remark}

Then, the identification problem of the switched linear system \eqref{s1} can be formulated as follows:
\begin{problem} \label{prob_1}
    Given a collection of input-output time series data $\mathcal{D}=\{(u_k,y_k)_{k=1}^N\}$ generated by the switched linear system \eqref{s1}, estimate the number of the subsystems $S$, the parameters of each subsystem $\{\theta_i\}_{i=1}^S$ and the discrete states $\{\lambda_k\}_{k=1}^N$.
\end{problem}
With a given nonnegative cost function, $l(\cdot)$, to measure the fitting errors of the identified model \eqref{s1}, the above identification problem can be written as the following mixed integer programming problem:
\begin{equation} \label{mixedipp}
    \begin{aligned}
    \min \limits_{S,\{\theta_i\}_{i=1}^S, \{\lambda_k\}_{k=1}^N} & \quad \sum_{k=1}^N l(y_k-\theta_{\lambda_k}^\top x_k) \\[1mm]
    \text{subj. to} &\quad \lambda_k\in \{1,\cdots,S\} ~\forall k=1,\cdots,N,  \\
    & ~S\in \mathbb{N},~\theta_i \in \mathbb{R}^n,~ \forall i=1,\cdots,S.
    \end{aligned}
\end{equation}
With the definition of the regressor vector $x_k$ in \eqref{s1}, we consider the vector of the bias:
\begin{equation} \label{phi}
    \phi(\theta)=\left[
        \begin{array}{ccc}
            \bar{x}_1^\top \bar{\theta} \\
            \vdots \\
            \bar{x}_N^\top \bar{\theta}
        \end{array}
    \right] = \bm{y}-X\top \theta - \bm{e}
\end{equation}
where $\bar{x}_k=[y_k, -x_k^\top]^\top, \bar{\theta}_k=[1,\theta^\top]^\top \in \mathbb{R}^{n+1}, k=1,\cdots,N$, $X=[x_1,\cdots,x_N]\in \mathbb{R}^{n\times N}$, $\bm{y}=[y_1,\cdots,y_N]^\top$ and $\bm{e}=[e_1,\cdots,e_N]^\top$. 

Inspired by \cite{bako2011identification}, the linear equation of the identification problem with an orthogonal projection matrix $A=I_N-X^\top(X\top X)^{-1}X$ ($I_N$ represents the $N$ identity matrix) can be formulated as follows:
\begin{equation} \label{linear_eq}
    Az=b
\end{equation}
where $\bm{z}=\phi(\theta)$ and $b=A\bm{y}$.
Considering that the matrix $A\in\mathbb{R}^{N\times N}$ can be replaced by the full row rank matrix $A_X\in \mathbb{R}^{(N-n)\times N}$ spanning the orthogonal complement of the column space of $X$.
Then the Identification Problem \ref{prob_1} can be regarded as the sparse optimization problem \eqref{P_J} with $J(z)=\sum_{k=1}^N l(y_k-\theta_{\lambda_k}^\top x_k)$ and $A_X z=A_X\bm{y}=b_X$.
For convenience, $\tau(X)$ is introduced to replace the mutual coherence $\mu(A_X)$ inspired by \cite{bako2011identification}, 
\begin{equation}\label{tau}
    \begin{aligned} 
        \tau(X) &= \max_{{1\leq t,k\leq N},~{t\neq k}} \frac{\lvert\Gamma_X(t,k)\rvert}{\sqrt{(1-\Gamma_X(t,t))(1-\Gamma_X(k,k))}} \\
        \Gamma_X &= X^\top (XX^\top)^{-1}X  
    \end{aligned}
\end{equation}
where $\Gamma_X(t,k)$ represents the value of the $(t, k)$-entry of $\Gamma_X$. According to the properties of idempotent matrix $A$, it follows $\mu(A_X) = \tau(X)$.

\subsection{Persistence of excitation problem}
The convergence of the identification algorithm always requires the input signal or the regressor vector constructed from input-output data to be persistently exciting \cite{narendra1987persistent}. 
The persistence of excitation is given as follows.
\begin{definition} \label{defpe}
    $\{x_k\}_{k=1}^N$ is a persistently exciting sequence of order $o_e \geq \dim(x_k)=n$ if there exist $\rho_1,\rho_2>0$ such that
    \begin{equation} \label{pe}
        \rho_1 I_n \leq \frac{1}{o_e}\sum_{k=0}^{o_e-1} x_{k_0+k}^\top x_{k_0+k} \leq \rho_2 I_{n}, \forall k_0\in \mathbb{N}
    \end{equation} 
    where $I_n$ denotes the $n$ identity matrix.
\end{definition}

In the absence of noise, $e_k\equiv 0$, the persistence of excitation problem of the switched linear system \eqref{s1} becomes whether the following equation has a unique solution for $\{\theta_i,\xi_{i,k},k=1,\cdots,N\}_{i=1}^S$ \cite{mu2022persistence},
\begin{equation} \label{pe_prob}
    \begin{aligned}
    &\sum_{k=1}^{N}\sum_{i=1}^{S} \xi_{i,k}(y_k-x_k^\top\theta_i)^2=0,\\
    \text{subject to} ~& \xi_{i,k} = I\{\lambda_k=i\} \\
    &\sum_{i=1}^{S}\xi_{i,k} = 1
    \end{aligned}
\end{equation}
where $I\{\cdot\}$ represents the indicator function.
Then for switched linear system \eqref{s1}, a PE condition was proposed based on Assumption \ref{assumption_ori} as follows \cite{mu2022persistence}.
\begin{assumption} \label{pe_ori}
    (See Assumption 1 in \cite{mu2022persistence}) For switched linear systems \eqref{s1}:
    \begin{enumerate}
        \item $\theta_i \neq \theta_j,~i\neq j,~\forall i,j \in \{1,\cdots,S\}$
        \item Assumption \ref{assumption_ori} holds.
        \item There exists an ordered set sequence $(\kappa_{p_1},\cdots,\kappa_{p_S})$ in the order of $(p_1,\cdots,p_S)$ with $(p_1,\cdots,p_S)$ being a permutation of $(1,\cdots,S)$ such that the following statements hold sequentially:
        \begin{enumerate}
            \item For any non-overlapping partition having the form $\kappa_{p_1}=\cup_{i=1}^S \kappa_{p_1}^{(S,i)}$, there exists some subset $\kappa_{p_1}^{s,i_1}$ with $i_1\in\{1,\cdots,S\}$ such that $\sum_{k\in\kappa_{p_1}^{S,i_1}}x_kx_k^\top$ is nonsingular.
            \item For any non-overlapping partition of $\kappa_{p_r}$ having the following form with $r\in\{2,\cdots,S\}$
            \begin{equation}
                \{\kappa_{p_r}^{(S-r+1,i)},i\in\{1,\cdots,S\}\backslash\{i_1,\cdots,i_{r_1}\}\},
            \end{equation}
            there exists $i_r \in \{1,\cdots,S\}\backslash\{i_1,\cdots,i_{r_1}\}$ such that $\sum_{k\in\kappa_{p_r}^{S-r+1,i_r}}x_kx_k^\top$ is nonsingular.
            \item $\sum_{k\in\kappa_{p_S}}x_kx_k^\top$ is nonsingular.
        \end{enumerate}
    \end{enumerate}
\end{assumption}

\section{Identification based on dynamic programming and sparsity inducing} 
In this section, a two-stage identification method for switched linear systems \eqref{s1} is proposed. In the first stage, the switching instants are identified based on an iterative dynamic programming algorithm in our prior work \cite{zheng2023identification} and the computational complexity of the algorithm has been further simplified. Under the segmented data which are divided according to the estimated switching instants, the parameters of each subsystem are determined based on the sparsity inducing methods.

\subsection{Constrained switching mechanism and problem formulation}
We begin with the limitation of Assumption \ref{assumption_ori} for the identification of switched linear system \eqref{s1}.

Actually, if we use fitting error to distinguish the mode of the data at that moment $(x_k,y_k)$ as mentioned in most literature, it is also known that $[y_k,-x_k^\top]^\top$ belongs to the kernel space (or the null space) of the parameter vector $[1,\theta_i^\top]$ which is written as, $[y_k,-x_k^\top]^\top \in \ker([-1,\theta_i^\top])$. Obviously, the data that belongs to the intersection of the two kernel spaces, $\ker([1,\theta_i^\top])\cap\ker([1,\theta_j^\top]),i\neq j$, cannot be distinguished whether it belongs to the $i$th or $j$th submodel. In the absence of noise, partitioning the data sample $(x_k,y_k)$ into the $i$th submodel or $j$th submodel does not affect the identification of the parameters $\{\theta_i\}_{i=1}^S$. While in the presence of noise, we cannot distinguish whether the distance from a sample point, $(x_k,y_k)$, to the kernel space, $\ker([1,\theta_i^\top])$, is influenced by noise or not and this leads to the inevitable identification errors and the influence on the unbiasedness of the identification. We make a graphical illustration for the result of the mismatch estimated partition in Fig. \ref{fig:mismatch}. 
\begin{figure} 
    \centering
    \includegraphics[width=\linewidth]{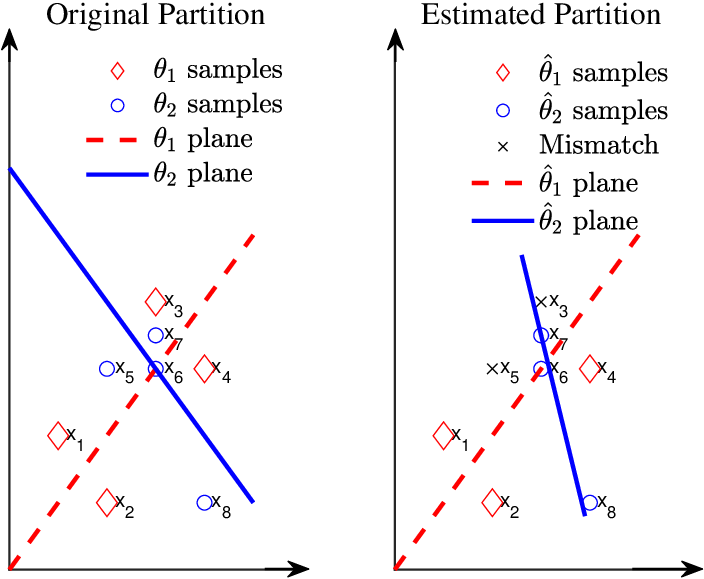}  
    \caption{A noisy example with data samples $\{x_k\}_{k=1}^8$ for the illustration of the mismatch in data partition. The original data partition corresponding to the true switching instants, $\{\lambda_k\}_{k=1}^8$, is shown on the left graph while the estimated partition with mismatch samples, $\{\hat{\lambda}_k\}_{k=1}^8$, is on the right. 
    } \label{fig:mismatch}
\end{figure}
As depicted in Fig. \ref{fig:mismatch}, $\{x_k\}_{k=1}^4$ and $\{x_k\}_{k=5}^8$ belong to different subsystems. If we only use the fitting errors to identify the switched linear system with $\{x_k\}_{k=1}^8$, the estimated partition with the least fitting error can be shown on the right graph in Fig. \ref{fig:mismatch}, with $\{x_1,x_2,x_4,x_5\} $ and $\{x_3,x_6,x_7,x_8\}$, thereby introducing bias in the identification of the parameter $\hat{\theta}_2$.  

However, by incorporating prior knowledge about the switching instants in time, specifically that mode switching only occurs once, the identification result is clearly formulated as the true partition with two segments: ${x_k}{k=1}^4$ and ${x_k}{k=5}^8$.
Therefore, methods that only rely on fitting errors for data partitioning depend on the noise distribution, while method leveraging the prior knowledge about switching instants in time may enhance the robustness of the identification.
This explains why the aforementioned methods based on Assumption \ref{assumption_ori} exhibit poor robustness against noise in practical applications.

To introduce the prior knowledge about switching instants in time, we need the segmentation of time series data before the identification for switched linear systems.
Then we propose the constrained switching mechanism which assumes the existence of a minimal interval between adjacent switching instants and stated as follows:
\begin{assumption} \label{assumption_new}
    Assumed that there exist $M$ different switching instants, $\{s_{m}\}_{m=1}^{M+1}, s_1=1$ in the collected time series data. The configuration of adjacent switching instants satisfies
    \begin{equation} \label{delta_s}
        s_{m+1}-s_{m} \geq \Delta s_{\min},
    \end{equation}
    where $\Delta s_{\min}$ is an integer and regulates the minimum distance between adjacent segmentation points.
\end{assumption}
Here we make a comment on Assumption \ref{assumption_new}, which not only avoid the nontrivial solutions in our proposed identification methods but also holds significant practical implications. 
This is because in practical applications, such as industrial processes, there exists a time interval between the switching of system modes due to the change of the environment or system operating points. 
This is because in industrial processes, the switching instants of systems always stem from changes in the environment and adjustments in operating points. Therefore, system switches cannot occur frequently.
This accounts for the prior knowledge of the constrained switching mechanism as described in Assumption \ref{assumption_new} on the timescale.

Next, we will introduce our two-stage identification framework for switched linear systems based on Assumption \ref{assumption_new}.
The identification problem \eqref{mixedipp} of the switched linear system \eqref{s1} are 
transformed into the identification of the switching instants $\{s_{m}\}_{m=1}^{M+1}$ and the identification of the parameters $\{\theta_i\}_{i=1}^S, \{\lambda_k\}_{k=1}^S$.
The qualitative trend analysis \cite{zhou2017new} and curve fitting \cite{bellman1969curve,bellman1961approximation,shawe2004kernel} are introduced for the identification of switching instants $\{s_{m}\}_{m=1}^{M+1}$, since the parameter vectors $\theta_{\lambda_k}$ can be considered as the trend of the segment, $\beta_m$. Then the identification problem for the switching instants are formulated as, 
\begin{equation} \label{problem_s1}
    \begin{aligned}
    \min \limits_{M, \{s_{m}\}_{m=1}^{M},\{\beta_m\}_{m=1}^M} & \quad \sum_{m=1}^{M+1}\sum_{k=s_m}^{s_{m+1}-1} l(y_k-\beta_{m}^\top x_k) \\[1mm]
    \text{subj. to} &~M\in \mathbb{N},~ s_1=1,~ s_{M+1}=N+1  \\
    & ~\beta_m \in \mathbb{R}^n,~ \forall m=1,\cdots,M.
    \end{aligned}
\end{equation}
If the data segmentation is sufficiently precise, we can make the following assumption.
\begin{assumption} \label{assumption_seg}
    Each segment of data belongs to the same subsystem. That is, $\forall m = 1,\cdots,M$, we have
    \begin{equation} \label{def_Lambda}
        \lambda_k=\Lambda_m,~k=s_{m},s_{m}+1,\cdots,s_{m+1}-1
    \end{equation}
    where $\Lambda_m\in \mathbb{N}$ indicates the order of the active subsystem among the time instants $s_m$ and $s_{m+1}$.
\end{assumption}
\begin{remark} \label{remark_seg}
    Assumption \ref{assumption_seg} is based on the segmented data by introducing the prior knowledge about the switching instants in time, and it degenerates into Assumption \ref{assumption_ori}, when there is only one data sample in each segment, denoted as $M=N$, $s_m=m$, $\forall m=1,\cdots,M$ and $\Delta{s_{min}}=1$.
\end{remark}


Considering that data is segmented by $\{s_{m}\}_{m=1}^{M+1}$ and Assumption \ref{assumption_seg} holds, the identification of parameters $\{\theta_i\}_{i=1}^S, \{\lambda_k\}_{k=1}^S$ can be formulated as
\begin{equation} \label{problem_s2}
    \begin{aligned}
    \min \limits_{S, \{\theta_i\}_{i=1}^S, \{\Lambda_m\}_{m=1}^M} & \quad \sum_{m=1}^{M}\sum_{k=s_m}^{s_{m+1}-1} l(y_k-\theta_{\Lambda_m}^\top x_k) \\[1mm]
    \text{subj. to} & ~S\in \mathbb{N},~\theta_i \in \mathbb{R}^n,~ \forall i=1,\cdots,S,\\
    &~ \quad \Lambda_m\in \{1,\cdots,S\} ~\forall m=1,\cdots,M.  \\
    \end{aligned}
\end{equation}
Next, we will introduce the two-stage identification method to address the identification problems of switched linear system \eqref{s1}, which is divided into the optimization problems \eqref{problem_s1} and \eqref{problem_s2}.


\subsection{Identification of the switching instants by dynamic programming}
Before introducing the algorithm, the expressions of the data segmentation are defined for the sake of convenience.
Suppose that the input-output time series data $\mathcal{T}=\{x_k\}_{k=1}^N,~\bm{y}=\{y_1,\cdots,y_N\}$ are collected by the switched linear system \eqref{s1} and there exist $M$ switching instants. Identification of the switching instants is to divide the time series data into $M$ non-overlapping data $\mathcal{T}=\cup_{m=1}^M \mathcal{T}_{m},~\bm{y}=\cup_{m=1}^M \bm{y}^m, ~\bm{e}=\cup_{m=1}^M \bm{e}^m$, where each segment is expressed as:
\begin{equation} \label{segment1}
    \begin{aligned}
        & \mathcal{T}_m = \left \{ x_{s_m}, x_{s_m +1}, \cdots, x_{s_{m+1} -1} \right \}\\
        & \bm{y}^m = \left \{ y_{s_m}, y_{s_m +1}, \cdots, y_{s_{m+1} -1} \right \} \\
        & \bm{e}^m = \left \{ e_{s_m}, e_{s_m +1}, \cdots, e_{s_{m+1} -1} \right \}.
    \end{aligned}
\end{equation}
Here, a grid of segmentation points $\{ s_m \}_{m=1}^{M+1}$ is introduced that satisfy $1=s_1 < s_2< \cdots < s_{M} <  s_{M+1} = N+1$, where $s_m$ indicates the index of the first data sample in $\mathcal{T}_m$, $m=1,\cdots,M$. 
For convenience, the restriction of entire time series $\mathcal{T}=\{x_k,~k=1,\cdots,N\}$ and $\bm{y}=\{y_k,~k=1,\cdots,N\}$ to a specific interval $[k,k'],k<k'$ is denoted as:
\begin{equation} \label{segment2}
    \begin{aligned}
        & \bm{y}(k:k') = [y_{k},y_{k+1},\cdots,y_{k'}]^\top\\
        &\mathcal{T}({k:k'}) =  [x_{k},x_{k+1},\cdots,x_{k'} ].\\
    \end{aligned}
\end{equation}

Then with the monotonic increasing $\ell_2$-norm function $l(\cdot)=\lVert \cdot \rVert_2$, the objective function of the mixed integer programming problem \eqref{problem_s1} with known number of switching instants $M$ can be rewritten as 
\begin{equation} \label{jm}
    \begin{aligned}
        \mathcal{J}(M) = \min_{ \{s_{m+1},\beta_m\}_{m=1}^{M}} \sum_{m=1}^M \lVert \bm{y}(s_m:s_{m+1})- \mathcal{T}_m^\top \beta_m \rVert_2^2 
    \end{aligned}
\end{equation}
where $\{s_j\}_{j=1}^{M+1}, s_1=1$ represent the switching instants. According to the least squares algorithm for the parameter estimation of linear systems, $\beta_m$ can be constructed from the segmented data $\mathcal{T}_m$ and $\bm{y}(s_m:s_{m+1})$, denoted as $\beta_m = (\mathcal{T}_m \mathcal{T}_m^\top)^{-1}\mathcal{T}_m \bm{y}(s_m:s_{m+1})$. Therefore, the problem \eqref{jm} will be addressed with the decisive parameters $\{s_{m}\}_{m=1}^M$ of $\mathcal{J}(M)$ by an iterative dynamic programming method in \cite{zheng2023identification}.

For a subproblem involving optimally dividing a shorter data trajectory $\mathcal{T}({1:k})$, where $k \leq N$, into $m$ segments, the total fitting error and the optimal segmentation choice are defined as follows:
\begin{equation} \label{cij}
    \begin{aligned}
        C_{m}(k) &= \min_{s_1, \cdots, s_{m}} \sum_{i=1}^{m} \lVert \bm{y}(s_i:s_{i+1})- \mathcal{T}_i^\top \beta_i \rVert_2^2 , \\
        \mathcal{K}_{m}(k)& = \mathop{\arg\min}_{s_1, \cdots, s_{m}} \sum_{i=1}^{m} \lVert \bm{y}(s_i:s_{i+1})- \mathcal{T}_i^\top \beta_i \rVert_2^2, \\
        & = \{\hat{s}_1^{m}(k), \cdots, \hat{s}_{m}^{m}(k)\}
    \end{aligned} 
\end{equation}
where $\hat{s}_1^{m}(k)=s_1=1$ is enforced and $s_{m+1}=k$ is assumed to simplify notations. Besides, we define $P_{m}(k)=\hat{s}_{m}^{m}(k)$. Obviously, the optimal value of \eqref{jm} is given by $\mathcal{J}(M)=C_{M}(N)$ and the optimal choice of segmentation can be readily recovered from $P_{M}(N)$.
In virtue of dynamic programming, $C_{m}(k)$ and $P_{m}(k)$ can be conveniently calculated based on the following recursions \cite{zheng2023identification} with Assumption \ref{assumption_new}.
\begin{itemize}
    \item Step 1: Initialize $C_{1}(k),~k=1,\cdots,N$:
    \begin{equation} \label{a1}
        C_{1}(k) = 0,~ P_{1}(k) =  k.
    \end{equation} 
    \item Step 2: Calculate $C_{m}(k)$ and $P_{m}(k)$ for $m=2,\cdots,M$ by following the Bellman principle of optimality: 
    \begin{equation} \label{a2}
        \begin{aligned}
            C_{m}(k) =& \min_{(m-1)\Delta s_{\min}\leq j \leq k}  C_{m-1}(j)+ d_{m,k}(j), \\
            P_{m}(k) =& \mathop{\arg\min}_{(m-1)\Delta s_{\min}\leq j \leq k} C_{m-1}(j)+ d_{m,k}(j), \\
        \end{aligned} 
    \end{equation}
    where $\forall k = (m-1)\Delta s_{\min},(m-1)\Delta s_{\min}+1,\cdots,N$, $d_{m,k}(j) = \min_{\beta \in \mathbb{R}^{n}} \lVert \bm{y}(j+1:k)- \mathcal{T}({ j+1:k})^\top \beta \rVert_2^2$.
    \item Step 3: Recover $\{\hat{s}_m\}_{m=1}^{M+1}$ successively with $\hat{s}_{M+1} = N$:
    \begin{equation} \label{a3}
        \begin{aligned}
            &\hat{s}_M = P_{M}(\hat{s}_{M+1}),\dots,~\hat{s}_2=P_{2}(\hat{s}_3),~\hat{s}_1 = 1.
        \end{aligned}
    \end{equation}
\end{itemize}
During the implementation of the identification process for the switching instants, we utilize the following criterion to optimally determine the number of segments $M$ \cite{zaliapin2004multiscale,zheng2023identification}:
\begin{equation} \label{jmc}
    M^*=\arg\min_M \frac{\log(\mathcal{J}(1)/\mathcal{J}(M))}{M-1}.
\end{equation}
\begin{remark} \label{RLS}
    Here we make a comment on the calculation of the linear estimation problem $d_{m,k}(j)$. Although the problem can be directly solved using the batch least-squares algorithm, the recursive least-squares algorithm \cite{benesty2011recursive} is introduced to significantly reduce the complexity of the algorithm. Then, $d_{m,k}(j+1)$ can be derived in the following recursive equations for $j=(m-1)\Delta s_{min},\cdots,k$:
    \begin{equation}
        \left\{ 
        \begin{aligned}
            d_{m,k}(j+1) & = d_{m,k}(j) + \frac{[z_{j+1}-x_{j+1}^\top \bm{\hat{\beta}}_j^m]^2}{1+x_{j+1}^TQ_j^m x_{j+1}} \\
            K_{j+1}^m &=  Q_{j}^m x_{j+1}[1+x_{j+1}^\top Q_{j}^m x_{j+1}]^{-1} \\
            Q_{j+1}^m &= [I-K_{j}^m x_{j+1}^\top]^{-1} Q_{j}^m \\
            \bm{\hat{\beta}}_{j+1}^m &= \bm{\hat{\beta}}_{j}^m + K_{j}^m[y_{j+1}-x_{j+1}^\top \bm{\hat{\beta}}_{j}^m] \\
        \end{aligned}\right. .
    \end{equation}
\end{remark}
The procedure for implementing the identification of switching instants by dynamic programming can be outlined as Algorithm \ref{algo:si}.
\begin{algorithm}[ht]  
    \caption{Identification of Switching Instants} \label{algo:si}
    \hspace*{0.02in} {\bf Require:} \\ 
    \hspace*{0.02in} Time series dataset $\mathcal{T}=\{x_{1},x_{2},\cdots,x_{N}\},~\bm{y}=\{y_1,\cdots,y_N\}$ and hyperparameters $\Delta s_{\min}$.
    \begin{algorithmic}[1]
    \State Calculate $\{C_{m}(n)\}_{m=1}^M$ and $\{P_{m}(n)\}_{m=1}^{M}$ as per \eqref{a1} and \eqref{a2}.
    \State Decide $M$ using the criterion \eqref{jmc} and derive $\{\hat{s}_m\}_{m=1}^{M}$ with \eqref{a3}. 
    \State Derive $M$ segments $\mathcal{T}=\cup_{m=1}^{M} \mathcal{T}_{m}$ based on the switching instants $\{\hat{s}_m\}_{m=1}^{M}$.
    \State \Return Segmented data $\mathcal{T}=\cup_{m=1}^M \mathcal{T}_{m},~\bm{y}=\cup_{m=1}^M \bm{y}^m$.
    \end{algorithmic}
\end{algorithm}

\subsection{Identification for the switched linear system by sparsity inducing}
In the previous section, the switching instants of system \eqref{s1} are identified, and segmented data are derived, denoted as $\bar{D}=\{\bar{D}_1,\cdots,\bar{D}_M\},~\bar{D}_m=[\bm{y}^m,-\mathcal{T}_{m}^\top]\in \mathbb{R}^{p_m\times (n+1)}, m=1,\cdots,M$ where $p_m$ represents the number of the samples in $\bar{D}_m$. The following assumption can be clearly made throughout the paper according to the identification of the switching instants.
Based on Assumption \ref{assumption_seg}, the identification of the discrete states $\{\lambda_k\}_{k=1}^N$ in the switched linear system \eqref{s1} can be transformed into the identification of the discrete states for the segmented data, denoted as $\{\Lambda_m\}_{m=1}^M$, where $\Lambda_m\in\{1,\cdots,S\}$ indicates the order of the active subsystem among the time instants corresponding to $\bar{D}_m$.
Given that the number of data samples in the data segment $\mathcal{T}_m$, denoted as $p_m$, exceeds the minimal time interval $v_n(X)$, and considering $rank(\mathcal{T}_m)\geq n=\dim(\theta)$ as defined in \eqref{k_index}, it follows that the linear system of equations $\bm{y}^m = \mathcal{T}_m^\top \theta +\bm{e}^m$ possesses a unique solution for $\theta$. Then, if Assumption \ref{assumption_new} with $\Delta s_{min}>v_n(X)$ and Assumption \ref{assumption_seg} holds, an analogue of Assumption \ref{assumption_ori} can be formulated as follows.
\begin{assumption} \label{assumption_new2}
    The data samples of each segment $(\mathcal{T}_m,\bm{y}^m)$ belong to the same subsystem $\Lambda_m$ after the identification for the switching instants of the system \eqref{s1}.
    There is no segment $(\mathcal{T}_m,\bm{y}^m)$ that fits two different submodels of the switched linear system \eqref{s1}, i.e., $\bm{y}^m = \mathcal{T}_m^\top \theta_i+\bm{e}^m = \mathcal{T}_m^\top \theta_j +\bm{e}^m \Rightarrow i=j.$
\end{assumption}

Based on the segmented data $\bar{D}=\{\bar{D}_1,\cdots,\bar{D}_M\}$, the identification problem \eqref{problem_s2} of the switched linear system \eqref{s1} can be formulated as an analogue of the sparse optimization \eqref{P_J} with the regressor matrix $X$, the output vector $\bm{y}$ and $J(z)=l(z)=\lVert z\rVert_0$:
\begin{equation} \label{prob_zero_ori}
    \begin{aligned}
        \min_{z\in\mathbb{R}^{N}} &~ \sum_{m=1}^{M}\|\bar{D}_m\bar{\theta} \|_0=\sum_{m=1}^{M} \|z^m\|_0 \\
        \text{subject to } &~Az=b \\
    \end{aligned}
\end{equation} 
where $\bar{\theta}=[1,\theta^\top]^\top$, $z=[\bar{D}^{\top}_1,\cdots,\bar{D}^{\top}_M]^\top\bar{\theta} $, $z^m=\bar{D}_m\bar{\theta}$, $A=I_N-X^\top(X\top X)^{-1}X,~ b = A\bm{y}$.
However, the optimization problem \eqref{prob_zero_ori} does not take into account the crucial prior knowledge summarized as Assumption \ref{assumption_new2}. To incorporate this knowledge and transform the parameters into $\theta$, the identification problem can be formulated as
\begin{equation} \label{prob_zero}
    \begin{aligned}
        \min_{\theta\in\mathbb{R}^{n}} &~ \sum_{m=1}^{M}p_m I\{\|\bar{D}_m\bar{\theta} \|_0 \geq p_m\} \\
        \text{subject to } &~Az=b \\
    \end{aligned}
\end{equation} 
where $I\{\cdot\}$ is the indicator function and $p_m$ represents the number of the samples in $\bar{D}_m$.

If Assumption \ref{assumption_new2} holds and all the subsystems are sufficient excited within the input-output time series data $\mathcal{D}$, the solution to the sparse problem \eqref{prob_zero} with $J(z)=\lVert z\rVert_0$ is a parameter vector representing one of the subsystem of the system \eqref{s1} as the following lemma.
\begin{lemma} \label{lemma_zero}
    Assume that the number of the data which belong to the $i$th submodel $N_i>S v_n(X),~\forall i=1,\cdots,S$, where $S$ is the number of the submodels. If Assumption \ref{assumption_new2}  holds, then the solution to the problem \eqref{prob_zero} $\theta_{i_0}$ is the parameter of the $i_0$th subsystem that has generated the most number of data.
\end{lemma}
\begin{proof}
    The proof is actually an extension of Lemma 7 in \cite{bako2011identification}.
    Suppose by contradiction that the solution $\theta$ to \eqref{prob_zero} does not lie in $\{\theta_1,\cdots,\theta_S\}$. And if we let $I(\theta)=\{p_m:\bar{D}_m\bar{\theta}=0\}$, then $|I(\theta)|\geq N_i \geq s v_n(X)$, where $|I(\theta)|$ is the sum of elements belonging to $I(\theta)$.
    Assumed that $n_i$ denotes the number of data generated by the $i$th submodel whose indices are contained in $I(\theta)$. Then we have
    \begin{equation}
        |I(\theta)| = \sum_{i=1}^S n_i \geq S v_n(X).
    \end{equation}
    In other words, there exists an index $j$ such that $n_j \geq v_n(X)$. Such data vector forms a matrix $\mathcal{X}=[\bar{D}_{t_1^j}^T,\cdots,\bar{D}_{t_{o_j}^j}^T]^T \in \mathbb{R}^{n_j \times (n+1)}$ where $\sum_{k=1}^{o_j}p_{t_k^j} = n_j$. 
    Then, both $\bar{\theta}$ and $\bar{\theta}_j$ lie in the kernel space (or the null space) of $\mathcal{X}$. 
    Since the initial entries of $\bar{\theta}$ and $\bar{\theta}_j$ are both equal to $1$, and $n_j \geq v_n(X)$, it follows that ${\theta}= {\theta}_j$, leading to a contradiction with the initial thesis.
\end{proof}

Considering Theorem \ref{theorem_mu} and the definition \eqref{tau}, it follows that $\theta_{i_0}$ is the unique solution to the combinatorial $\ell_0$-norm optimization problem \eqref{prob_zero} if there is a solution $z = [\bar{\theta}^\top \bar{D}_1^\top,\cdots,\bar{\theta}^\top \bar{D}_M^\top]^\top$ obeying that,
\begin{equation} \label{condition_mua}
    \|z\|_0< \frac{1}{2}(1+\frac{1}{\mu(A)}).
\end{equation}

Although it is challenging to solve the $\ell_0$-norm optimization problem directly, such as $P_J$ with $J(z)=\|z\|_0$, the segmented data and prior knowledge contribute significantly to addressing this sparse problem. We improve the traditional orthogonal matching pursuit (OMP) algorithm \cite{tropp2007signal,wang2012generalized} to address the sparse problem \eqref{prob_zero} in Appendix \ref{append_1}. However, this $\ell_0$-norm based methods exhibit bad robustness against noise and have significant computational complexity \cite{wang2012generalized}. Based on the above reasons, we consider the $\ell_1$-norm relaxation of the combinatorial $\ell_0$-norm sparse optimization problem \eqref{prob_zero} with the segmented data,
\begin{equation} \label{prob_one}
    \begin{aligned}
        \min_{\theta\in\mathbb{R}^{n}} ~& \lVert W_X z \rVert_1 \\
        \text{subject to } ~&A_Xz=b_X \\
        & z = [\bar{\theta}^\top \bar{D}_1^\top,\cdots,\bar{\theta}^\top \bar{D}_M^\top]^\top\\
        & W_X = \text{diag}(\nu(A_X^1)I_{p_1},\cdots,\nu(A_X^M)I_{p_M}) \\
    \end{aligned}
\end{equation} 
where $A_X = [A^{1}_X,\cdots,A^{M}_X]$ is partitioned in column order according to the switching instants $\{ s_m \}_{m=1}^{M+1}$, $W_X$ is the weighting matrix for increasing the sparsity of $z$ and $\nu(A_X^m)$ is defined as the average of the Euclidean norms of the column vectors in matrix $A_X^m = [a_{s_m}, a_{s_m +1}, \cdots, a_{s_{m+1} -1}]$,
\begin{equation} \label{nu}
    \nu(A_X^m) = \frac{\sum_{k=s_m}^{s_{m+1}-1} \lVert a_k \rVert_2}{p_m}, m=1,\cdots,M.
\end{equation}
We have the following theorem to discuss the uniqueness of the solution, $\theta$, for the sparse problem \eqref{prob_one}.
\begin{theorem} \label{theorem_one}
    If Assumption \ref{assumption_new2} holds and there is a parameter vector $\theta$ achieving the error vector $z = [\bar{\theta}^\top \bar{D}_1^\top,\cdots,\bar{\theta}^\top \bar{D}_M^\top]^\top$ based on the segmented data $\bar{D}=\{\bar{D}_1,\cdots,\bar{D}_M\}$ such that
    \begin{equation} \label{tau2}
        \|z\|_0< \frac{1+\tau(X)}{2},
    \end{equation}
    then $\theta$ is the unique solution to the $\ell_1$-norm optimization problem. $\tau(X)$ is defined in \eqref{tau}.
\end{theorem}
\begin{proof}
    According to the definition \eqref{tau}, we have $\tau(X)=\mu(A)=\mu(A_X)$. Then, if we only focus on the variable $z$, the $\ell_0$-norm problem \eqref{prob_zero} and the $\ell_1$-norm problem \eqref{prob_one} have the unique solution $z$ with the condition \eqref{tau2} according to Theorem \ref{theorem_mu}. Furthermore, if $z$ is unique and Assumption \eqref{assumption_new2} holds, $\theta$ can be formulated uniquely based the $z = [\bar{\theta}^\top \bar{D}_1^\top,\cdots,\bar{\theta}^\top \bar{D}_M^\top]^\top$.
\end{proof}

The unbiasedness and uniqueness of the solution $\theta$ to the sparse problem have been guaranteed according to aforementioned Lemma \ref{lemma_zero} and Theorem \ref{theorem_one}.
Then by solving the sparse optimization problem \eqref{prob_one}, the parameter of one of the subsystem $\theta$ can be extracted, and we use the following criterion to determine whether the segments belong to the extracted parameter vector $\theta$ in presence of noise,
\begin{equation} \label{I_theta}
    \begin{aligned}
        I(\theta) =\{\bar{D}_m:\frac{\|\bar{D}_m\theta\|_1}{\nu(\bar{D}_m) \lVert \theta \rVert_2 p_m}  \leq \varepsilon_{thres}, m=1,\cdots,M\}
    \end{aligned}
\end{equation}
where $\nu(\cdot)$ is defined in \eqref{nu}.

Inspired by the reweighted $\ell_1$ minimization method in \cite{bako2011identification}, an $\ell_1$-norm optimization algorithm with momentum method \cite{hao2021adaptive} is proposed for the sparse optimization problem \eqref{prob_one}, detailed in Algorithm \ref{algo:one}.
\begin{algorithm}[ht]
    \caption{$\ell_1$-norm Optimization Algorithm} \label{algo:one}
    \hspace*{0.02in} {\bf Require:} \\ 
    \hspace*{0.02in} $A_X\in \mathbb{R}^{(N-n)\times N},~b_X\in \mathbb{R}^{N-n}$, switching instants $\{ s_m \}_{m=1}^{M+1}$, segmented data $\bar{D}=\{\bar{D}_1,\cdots,\bar{D}_M\}$ and hyperparameters $\varepsilon_0, \varepsilon_{thres} , \alpha, \eta, v_0$.
    \begin{algorithmic}[1]
    \State Initial settings: $j=0$, $v_0=0.1$
    \begin{itemize}
        \item Initial weights: $W^{(0)}=diag(w_1^{(0)} I_{p_1},\cdots,w_M^{(0)} I_{p_M})$ with $w_k^{(0)}=1,k=1,\cdots,N$
    \end{itemize}
    \State While $\|\theta^j-\theta^{(j-1)}\|_2 \le \varepsilon_0$ 
    \begin{itemize}
        \item Solve the convex $\ell_1$-norm optimization problem
        \begin{equation}
            \begin{aligned}
                \theta^j = \arg\min_{\theta\mathbb{R}^{n}} \lVert W^{(j)} W_L^{-1} \phi^0(\theta) \rVert_2
            \end{aligned}
        \end{equation}
        where $\phi^0(\theta) = [\bar{\theta}^\top \bar{D}_1^\top,\cdots,\bar{\theta}^\top \bar{D}_M^\top]^\top$, $W_L = \text{diag}(\nu(\bar{D}_1)I_{p_1},\cdots,\nu(\bar{D}_M)I_{p_M})$.
        \item Update the weights with \eqref{nu}
        \begin{equation}
            \begin{aligned}
                v_{j+1} & = \alpha v_j - \eta \frac{1}{\nu(\bar{D}_m \bar{\theta}^j)}\\
                w_m^{(j+1)} &= \frac{1}{ \nu(\bar{D}_m \bar{\theta}^j)+v_{j+1}},~m=1,\cdots,M \\
            \end{aligned}
        \end{equation} 
        \item $j=j+1$
    \end{itemize}
    \State Concatenate the data, 
    \begin{equation}
        \begin{aligned}
            I(\theta^j) =\{\bar{D}_m:\frac{|\bar{D}_m\theta^j|}{\nu(\bar{D}_m) \lVert \theta^j \rVert_2}  \leq \varepsilon_{thres}, m=1,\cdots,M\}, 
        \end{aligned}
    \end{equation}
    and re-estimate the parameter vector $\theta^j$ as \eqref{thetar}.
    \State \Return $\theta^j,~I(\theta^j)$ 
    \end{algorithmic}
\end{algorithm}

\subsection{The complete algorithm}

Considering Theorem \ref{theorem_one} and the definition of \eqref{tau}, all the parameter vectors $\{\theta_1,\cdots,\theta_s\}$ can be derived by solving the sparse problem \eqref{prob_zero} as Algorithm \ref{algo:one}. If one of the subsystem parameters $\theta$ is extracted, the data segments belonging to this subsystem are removed from the data $[\bar{\theta}^\top \bar{D}_1^\top,\cdots,\bar{\theta}^\top \bar{D}_M^\top]^\top$ and $z$ can be reformulated based on the remaining data. Obviously, if we want to extract the parameter vectors $\{\theta_1,\cdots,\theta_s\}$ one after another uniquely, the condition \eqref{tau2} still need to be satisfied by the remaining data according to Theorem \ref{theorem_one}. 

Then we give the complete condition for the unbiasedness and uniqueness of $\{\theta_1,\cdots,\theta_s\}$ in the following theorem, which is mentioned in \cite{bako2011identification}.
\begin{theorem} \label{theorem_all}
    (Theorem $14$ in \cite{bako2011identification}) Consider the data matrix $X\in \mathbb{R}^{n \times N }$ generated by the system \eqref{s1} and assume that 
    \begin{equation}
        \begin{aligned}
            &N_1>N-\vartheta(X_1) >0, \\
            &N_2>N-N_1-\vartheta(X_2)>0,\\
            &\vdots\\
            &N_{S-1}>N-N_1-\cdots-N_{S-2}-\vartheta(X_{S-1})>0 \\
        \end{aligned}
    \end{equation}
    where $\vartheta(X)= 1/2(1+1/\tau(X))$ with $\tau(X)$ defined in \eqref{tau}. Then all the parameter vectors $\{\theta_1,\cdots,\theta_S\}$ can be extracted one after another by solving the sparse problem \eqref{prob_zero}.
\end{theorem}

Therefore, the complete framework for the identification of the switched linear systems \eqref{s1}, including the identification of switching instants, data segmentation and the identification of subsystem parameters, can be formulated as Algorithm \ref{algo:com}.
\begin{algorithm}[h] 
    \caption{Identification for switched linear system with $\ell_1$-norm} \label{algo:com}
    \hspace*{0.02in} {\bf Require:} \\ 
    \hspace*{0.02in} Time series dataset $\mathcal{T}=\{x_{1},x_{2},\cdots,x_{N}\},~\bm{y}=\{y_1,\cdots,y_N\}$ and hyperparameters $\Delta s_{\min}, \varepsilon_0, \varepsilon_{thres} , \alpha, \eta, v_0$.
    \begin{algorithmic}[1] 
    \State Determine the switching instants $\{ s_m \}_{m=1}^{M+1}$ and derive the segmented data $\bar{D}=\{\bar{D}_1,\cdots,\bar{D}_M\}$ by Algorithm \ref{algo:si} with the hyperparameter $\Delta s_{\min}$.
    \State  Initial settings: $S=0,~\Theta=\{\null\},~E=\bar{D}$
    \State While $E\neq \emptyset$
    \begin{itemize}
        \item Based on the rest dataset $E$, derive the parameter vector $\theta$ and its segment data $I(\theta)$ through the $\ell_1$-norm problem \eqref{prob_one} by Algorithm \ref{algo:one} with $\varepsilon_0, \varepsilon_{thres} , \alpha, \eta, v_0$.
        \item Update the set: 
        $E \leftarrow  E \cap I(\theta), \Theta \leftarrow \Theta \cup \{\theta\}$.
        \item $S=S+1$
    \end{itemize}
    \State \Return $\Theta$ and $S$.
    \end{algorithmic}
\end{algorithm}

\section{A new sufficient persistent excitation condition for switched linear system}
In this section, we propose a new sufficient persistent excitation condition for switched linear system, based on the constrained switching mechanism, as a comparison with Assumption \ref{assumption_ori}. 
The difference lies in the fact that in the previous section, we proposed Assumption \ref{assumption_new} with $\Delta{s_{min}}>v_n(X)$ from the perspective of sparse optimization to address the identification problem, while in this section, we illustrate the sufficiency of the persistent excitation condition from an existential standpoint.

According to the data segmentation $\bar{D}=\{\bar{D}_1,\cdots,\bar{D}_M\}$, if Assumption \ref{assumption_seg} holds, the PE problem \eqref{pe_prob} can be formulated as
\begin{equation} \label{pe_prob_new}
    \begin{aligned}
    &\sum_{m=1}^{M}\sum_{i=1}^{S} \zeta_{i,m}\|\bar{D}_m\bar{\theta}_i\|_2^2=0,\\
    \text{subject to} ~& \zeta_{i,m} = I\{\Lambda_m=i\} \\
    &\sum_{i=1}^{S}\zeta_{i,m} = 1
    \end{aligned}
\end{equation}
where $\Lambda_m\in\{1,\cdots,S\}$ indicates the order of the active subsystem among the time instants corresponding to $\bar{D}_m$ and $I\{\cdot\}$ represents the indicator function.
To discuss the uniqueness of the solution to the problem \eqref{pe_prob_new}, the identification for the switching instants and the data segmentation are further investigated. 
Considering the segmentation with $e_k\equiv 0$, we have $\beta_m=\theta_{\Lambda_m}$. Then how to guarantee the identifiability of the switching instants or the uniqueness of the solution to the segmentation problem \eqref{jm} can be quite challenge and confusing. 

Suppose that the adjacent two segments $ \bar{D}_m, \bar{D}_{m+1} $ are partitioned with the switching instant $s_{m+1}$. Then, if there exist some samples fitting two subsystems, $x_{s_{m+1}-j}^\top\theta_{\Lambda_m}=x_{s_{m+1}-j}^\top \theta_{\Lambda_{m+1}},~\forall j=1,\cdots,k'$, then $s_{m+1}'=s_{m+1}-j,j=1,\cdots,k'$ can be another solution to the segmentation problem \eqref{jm} for $s_{m+1}$. 
This example is provided to illustrate the challenge of identifiability of the switching instants, and to avoid the ambiguity, an assumption for the uniqueness of the solution to the segmentation problem is proposed as Assumption \ref{assumption_si}.

\begin{assumption} \label{assumption_si}
    Consider the data segmentation, $\{ s_m \}_{m=1}^{M+1}$, $\{\beta_m\}_{m=1}^M$ and $\bar{D}=\{\bar{D}_1,\cdots,\bar{D}_M\}$, for $m=1,\cdots,M$:
    \begin{enumerate}
        \item $\beta_m\in\{\theta_1,\cdots,\theta_S\}$
        \item $\beta_m\neq \beta_{m+1}, m=1,\cdots,M-1$.
        \item $x_{s_m}^\top \beta_m=0, x_{s_m}^\top \beta_{m'}\neq 0, \forall m'=\{1,\cdots,M\}\backslash\{m\}$
        \item There exists a subset $ \bar{D}_m' \subseteq \bar{D}_m, \forall m=1,\cdots,M $ such that $\sum_{x_k\in\bar{D}_m'}x_kx_k^\top$ is nonsingular.
    \end{enumerate}
\end{assumption}
\begin{theorem} \label{theorem_seg}
    In the absence of noise, $e_k\equiv 0$, if Assumption \ref{assumption_si} holds, then $\{s_m\}_{m=1}^{M+1}$, $\{\beta_m\}_{m=1}^M$ is the unique solution to the problem \eqref{jm} with fixed $M$.
\end{theorem}
\begin{proof}
    We will prove the theorem by the Mathematical Induction and take $M$ as the induction variable.

    To the basic case, consider $M=2, \{s_m\}_{m=1}^3$ with $s_1=1,s_3=N+1$. Obviously, the case $M=2$ is true and the problem \eqref{jm} with $M=2$ have the unique solution according to the Assumption \ref{assumption_si}.

    Then we make the induction that the theorem holds for $M-1$ and consider the case with $M$. Suppose by contradiction that there are two solutions $\{s_m,\beta_m\}_{m=1}^{M}$ and $\{s_m',\beta_m'\}_{m=1}^{M}$ to the problem \eqref{jm} with $M$, and we have $s_1=s_1'=1,s_{M+1}=s_{M+1}'=N+1$. 
    According to the last segment and the Assumption \ref{assumption_si}, we have $\beta_{M-1}\neq\beta_M=\beta_M'\neq \beta_{M-1}'$.
    Considering the beginning of the last segment, $s_{M}'$, and the Assumption \ref{assumption_si} with $s_{M}$, $x_{s_{M}}^\top \beta_M=0, x_{s_{M}}^\top \beta_{M-1}\neq 0$, we have $s_{M}'\leq s_{M}$. 
    This is because if not, this mismatch will result in $x_{s_{M}}^\top \beta_{\bar{m}}'\neq 0, \bar{m}\in\{1,\cdots,M-1\}$ which is contradicted to the third condition in Assumption \ref{assumption_si}. For the same reason with $s_{M}$ and $x_{s_{M}'}^\top \beta_M=0, x_{s_{M}'}^\top \beta_{M-1}'\neq 0$, we have $s_{M}\leq s_{M}'$. Then we can derive $s_{M}= s_{M}'$ and the partition for the last data segment $\bar{D}_m$ is unique. According to the induction, the solution to the partition of the data $\bar{D}=\{\bar{D}_1,\cdots,\bar{D}_{M-1}\}$ with $M-1$ segments is unique and the case with $M$ is proved.
    Finally, according to the induction, the theorem is proved.
\end{proof}

Although the uniqueness of the solution to the segmentation problem \eqref{jm} and the identifiability of the switching instants are guaranteed by Theorem \ref{theorem_seg}, 
The uniqueness of the solution to the segmentation problem \eqref{jm} and the identifiability of the switching instants are not necessary for the identification of the parameters of the switched linear system \eqref{s1}, although they are guaranteed by Theorem \ref{theorem_seg}. As a matter of fact, Assumption \ref{assumption_si} is a sufficient but not necessary condition for Assumption \ref{assumption_new2}. Then, we continue to discuss the persistent excitation condition for the switched linear system based on Assumption \ref{assumption_new2}.

\begin{assumption} \label{new_pe}
    For the switched linear system \eqref{s1} with segmented data $\bar{D}=\{\bar{D}_1,\cdots,\bar{D}_M\}$ and switching instants $\{ s_m \}_{m=1}^{M+1}$:
    \begin{enumerate}
      \item $\theta_i \neq \theta_j,~i\neq j,~\forall i,j \in \{1,\cdots,S\}$.
      \item Assumption \ref{assumption_new2} holds.
      \item There exists a subset $ \bar{D}_m' \subseteq \bar{D}_m, \forall m=1,\cdots,M $ such that $\sum_{x_k\in\bar{D}_m'}x_kx_k^\top$ is nonsingular.
      \item For any segment $\bar{D}_m$, there exists $x_{k^m}\in \bar{D}_m$ such that $x_{k^m}^\top \theta_{\Lambda_m}\neq 0, \forall i\in\{1,\cdots,S\}\backslash\{\Lambda_m\}$.
    \end{enumerate}
\end{assumption}

\begin{theorem} \label{theorem_pe}
    In the absence of noise $e_k\equiv 0,~\forall k\in \mathbb{R}^N$ and suppose Assumption \ref{new_pe} holds. Then the segmented data $\bar{D}=\{\bar{D}_1,\cdots,\bar{D}_M\}$ and $\{\zeta_{i,m},i=1,\cdots,S\}_{m=1}^M$ are persistently exciting for the switched linear system \eqref{s1}.
\end{theorem}
\begin{proof}
    Considering each segment $\bar{D}_m$ for $m=1,\cdots,M$, we have the unique solution to the equation $\|\bar{D}_{m}\bar{\theta}\|_2^2=0$ with $\bar{\theta}=\bar{\theta}_{\Lambda_m}, \zeta_{\Lambda_m,m}=1,\zeta_{i,m}=0,\forall i\in\{1,\cdots,S\}\backslash\{\Lambda_m\}$ according to the fourth condition of Assumption \ref{assumption_new2}. Then ignoring the permutation of the subsystem order $(1,\cdots,S)$, $\{\theta_i,\zeta_{i,m},m=1,\cdots,M\}_{i=1}^{S}$ is the unique solution to the segmentation problem \eqref{pe_prob_new} and $\bar{D}=\{\bar{D}_1,\cdots,\bar{D}_M\}$ and $\{\zeta_{i,m},i=1,\cdots,S\}_{m=1}^M$ are persistently exciting.
\end{proof}

\section{Experiments}
In this section, experiments with both periodic switching instants and random switching instants are conducted to assess the performance of our identification methods for switched linear systems.
\subsection{Numerical Experiments with periodic switching instants}

In this section, our identification algorithms are applied into the switched linear system \eqref{s1} composed of three linear subsystems of order two $n_a=n_b=2,n=4$ which is similar in \cite{bako2011identification} with their parameters $\lambda_k\in\{1,2,3\}, k=1,\cdots,N$ and $\{\theta_i\}_{i=1}^3$:
\begin{equation}
    \begin{aligned}
        \theta_1 & = [-0.4,0.25,-0.15,0.08]^\top \\
        \theta_2 & = [0.55,-0.58,-1.1,1.2]^\top \\
        \theta_3 & = [1,-0.24,-0.65,0.3]^\top.
    \end{aligned}
\end{equation}
The time series input-output data $\mathcal{D}=\{(u_k,y_k)_{k=1}^N\}$ with $N=1000$ samples is generated by \eqref{s1} with the above parameters and the following conditions:
\begin{enumerate}
    \item The input $\{u_k\}_{k=1}^N$ is generated from Gaussian distribution with zero mean and variance unity.
    \item The Gaussian noise term $\{e_k\}_{k=1}^{N}$ with zero mean is generated such that the Signal to Noise Ratio (SNR) is equal to $30$dB with respect to the output $\{y_k\}_{k=1}^N$ by adjusting the variance. 
    \item The switching instants are periodic $\{s_m\}_{m=2}^{11}=\{100,200,\cdots,1000\}$ with the discrete state of each segment $\{\Lambda_m\}_{m=1}^{10}=\{1,2,3,1,2,1,2,1,3,2\}$.
\end{enumerate}
Although the input $\{u_k\}_{k=1}^N$ is generated randomly, the intersection of the kernel spaces corresponding to subsystem parameters $\theta_1, \theta_2, \theta_3$ are inevitable, which is summarized as follows:
\begin{equation}
    \begin{aligned}
        &|\{I(\theta_1)\cap I(\theta_2)\}|/N =  30.4\% \\
        &|\{I(\theta_1)\cap I(\theta_3)\}|/N =  26.9\% \\
        &|\{I(\theta_2)\cap I(\theta_3)\}|/N =  32.1\% \\
    \end{aligned}
\end{equation}
where $|\cdot|$ denotes the cardinality of the set and $I(\theta)$ is defined in \eqref{I_theta} with $\varepsilon_{thres}=10^{-4}$.
The first $800$ samples are defined as training set while the last $200$ samples are used for test set. 
A criterion named "Fit" is used to estimate the fitting error and the prediction accuracy which is proposed in \cite{ljung1995system}:
\begin{equation} \label{fit}
    \text{Fit} = \left(1-\frac{\lVert \hat{y}-y \Vert_2}{\lVert y-\bar{y} \bm{1}_N\Vert_2}\right) \times 100\%
\end{equation}
where $y$ is the true output, $\hat{y}$ is the model output, $\bar{y}$ is the mean of $y$ and $\bm{1}_N$ represents the $N$-dimensional column vector with all entries equal to one.

Then the reproducible algorithm in \cite{bako2011identification} is compared with our proposed Algorithm \ref{algo:com} which utilizes $\ell_1$-norm optimization, with hyperparameters set as outlined in Table \ref{para_set}. The prediction results are depicted in Fig. \ref{fig:pre1}. Additionally, the improved OMP algorithm in Appendix \ref{append_1} is employed to identify the segmented data as for further comparison. 
\begin{table}[hb]
    \begin{center}
    \caption{Parameter Settings} \label{para_set}
    \begin{tabular}{llll} \hline
    Parameters & Value \\\hline
    $\Delta_{s_{min}}$ & $10$\\
    $\varepsilon_0$ & $0.001$ \\
    $\varepsilon_{thres}$ & $1\times 10^{-4}$ \\
    $\alpha$ & $0.9$ \\
    $\eta$ & $0.01$ \\
    $v_0$ & $0.1$ \\ \hline
    \end{tabular}
    \end{center}
\end{table}
\begin{figure} 
    \centering
    \includegraphics[width=\linewidth]{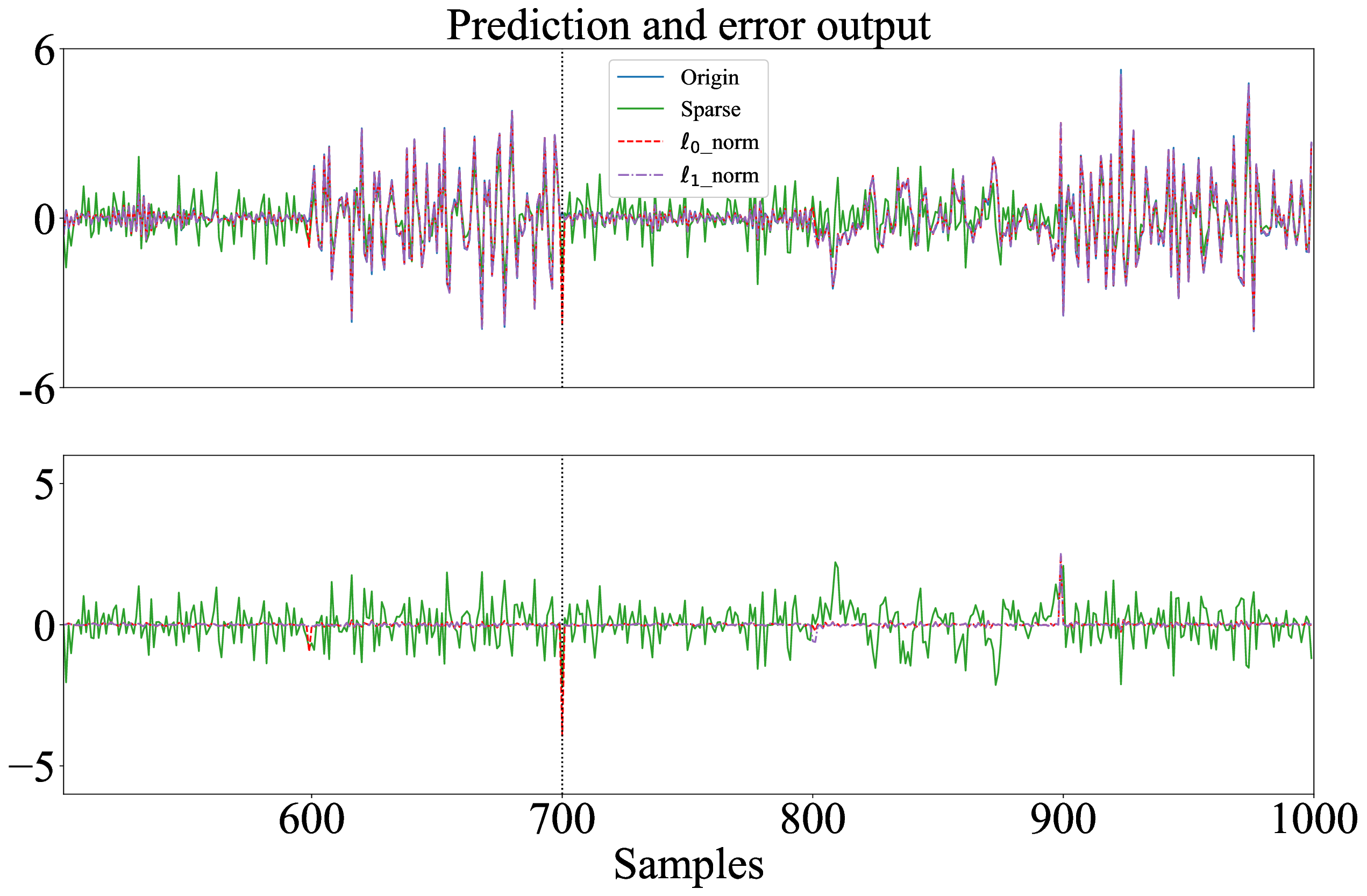}  
    \caption{300 identification and 200 test data with $SNR=30dB$. The blue solid line, marked with "Origin", represents the true output. The green solid line, marked with "Sparse", denotes the model prediction by the reproducible algorithm in \cite{bako2011identification}. 
    The red dotted line, marked as "$\ell_0$-norm" represents the model prediction by the improved OMP algorithm, as Algorithm \ref{algo:zero} in Appendix \ref{append_1} while the purple dot-dashed line, marked as "$\ell_1$-norm" represents the result of Algorithm \ref{algo:com} with $\ell_1$-norm. 
    } \label{fig:pre1}
\end{figure}

To test the robustness of the identification algorithm for the switched linear system, $100$ different independent runs are also conducted as the Monte Carlo experiments. 
The results for the single run and the Monte Carlo experiments are summarized in Table \ref{compare_fit} and Table \ref{monte_para}.
The histogram of the $100$ independent runs for Algorithm \ref{algo:com} with $\ell_0$-norm is depicted in Fig. \ref{fig:hist0} while the result for $\ell_1$-norm is shown in Fig. \ref{fig:hist1}.

As depicted in Fig. \ref{fig:pre1} and the summary in Table \ref{compare_fit} and Table \ref{monte_para}, illustrating the prediction results, both $\ell$-norm and $\ell_1$-norm methods in Algorithm \ref{algo:com} exhibit superior performance in reconstructing the true system output compared to the reproducible identification algorithm in \cite{bako2011identification} and Algorithm \ref{algo:com} based on the identification of switching instants exhibits better prediction accuracy in output fitting with periodic switching instants. Moreover, Algorithm \ref{algo:com} with $\ell_1$-norm exhibits stronger robustness against noise than $\ell_0$-norm.
\begin{table}
    \centering
    \caption{Comparison of the proposed algorithms with periodic switching instants} \label{compare_fit}
    \begin{tabular}{@{}lccc@{}}
        \toprule
        Algorithm & Sparse & $\ell_0$-norm & $\ell_1$-norm \\ \midrule
        Fit  (Single run in Fig. \ref{fig:pre1})       & 52.14\% & 85.17\% & 91.16\% \\
        Average Fit (Monte Carlo) & 47.02\% & 90.67\% & 92.51\% \\ \bottomrule
    \end{tabular}
\end{table}
\setlength{\tabcolsep}{4pt}
\begin{table}
    \centering
    \caption{Comparison of the proposed identification algorithms over $100$ independent runs} \label{monte_para}
    \begin{tabular}{@{}lccc@{}}
        \toprule
        Method & $\hat{\theta}_1$ & $\hat{\theta}_2$ & $\hat{\theta}_3$ \\ \midrule
        $\ell_0$ norm & 
        $\begin{aligned}
            -0.321 \pm 0.086 \\  0.166 \pm 0.075 \\ -0.260 \pm 0.089 \\  0.210 \pm 0.117
        \end{aligned} $ & 
        $\begin{aligned}
            0.505 \pm 0.018 \\ -0.468 \pm 0.072 \\ -0.986 \pm 0.055 \\ 0.992 \pm 0.145 \\
        \end{aligned} $  &
        $\begin{aligned}
            0.638 \pm 0.114 \\ -0.082 \pm 0.048 \\ -0.624 \pm 0.034\\  0.156 \pm 0.065 
        \end{aligned} $  \\ \hline
        $\ell_1$ norm & 
        $\begin{aligned}
            -0.368 \pm 0.043 \\  0.221 \pm 0.004 \\  -0.155 \pm 0.021 \\  0.086 \pm 0.001 
        \end{aligned} $ & 
        $\begin{aligned}
            0.550 \pm 0.000\\ -0.579 \pm 0.000  \\ -1.096 \pm 0.000 \\  1.191 \pm 0.001
        \end{aligned} $  &
        $\begin{aligned}
            0.832 \pm 0.006 \\  -0.182  \pm 0.047 \\ -0.648  \pm 0.000 \\ 0.302  \pm 0.024\\
        \end{aligned}$  \\ \bottomrule
    \end{tabular}
\end{table}
\begin{figure} 
    \centering
    \includegraphics[width=\linewidth]{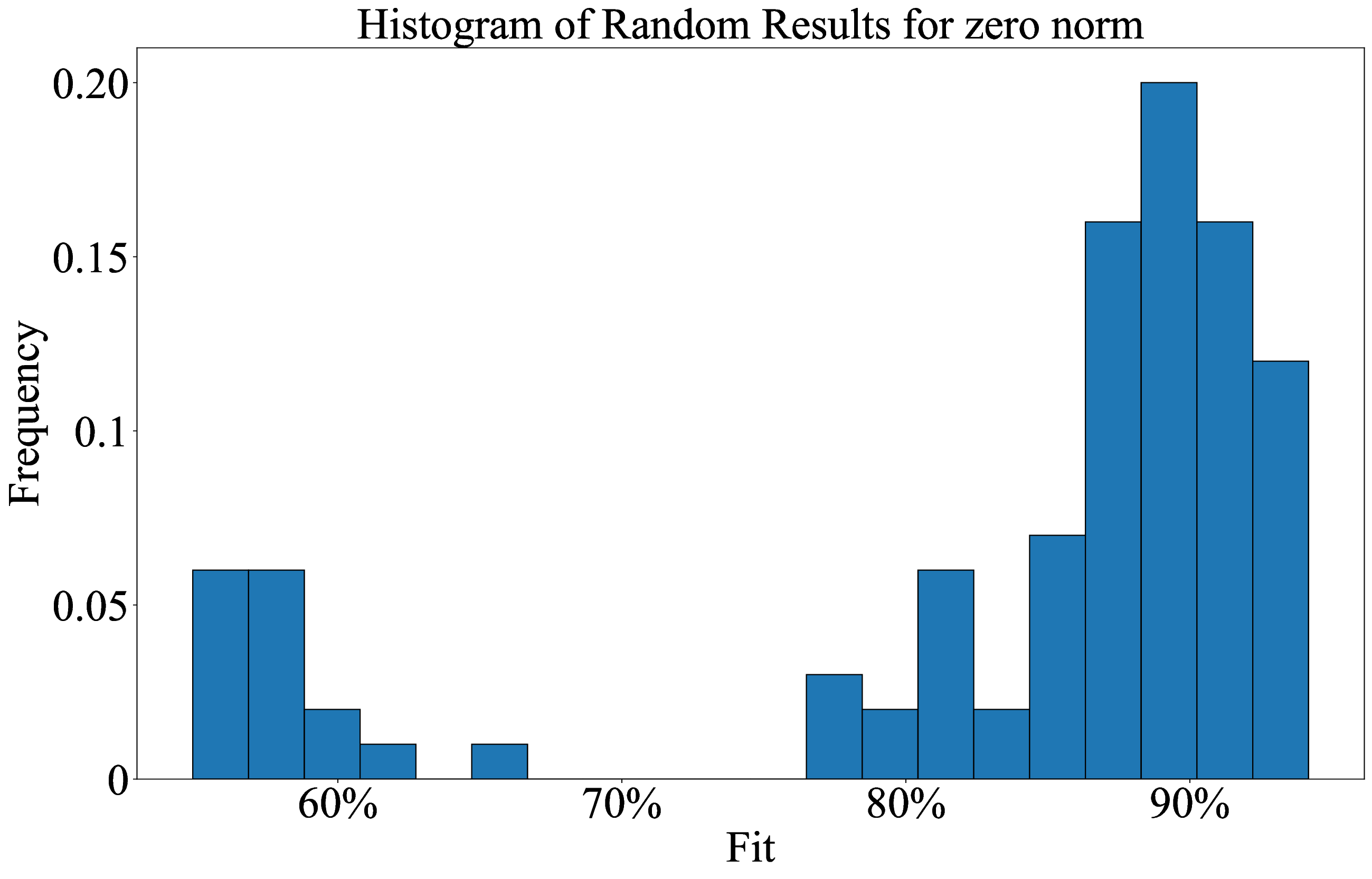} 
    \caption{The histogram of Fit for $100$ independent runs by Algorithm \ref{algo:com} with $\ell_0$-norm.} \label{fig:hist0}
\end{figure}
\begin{figure} 
    \centering
    \includegraphics[width=\linewidth]{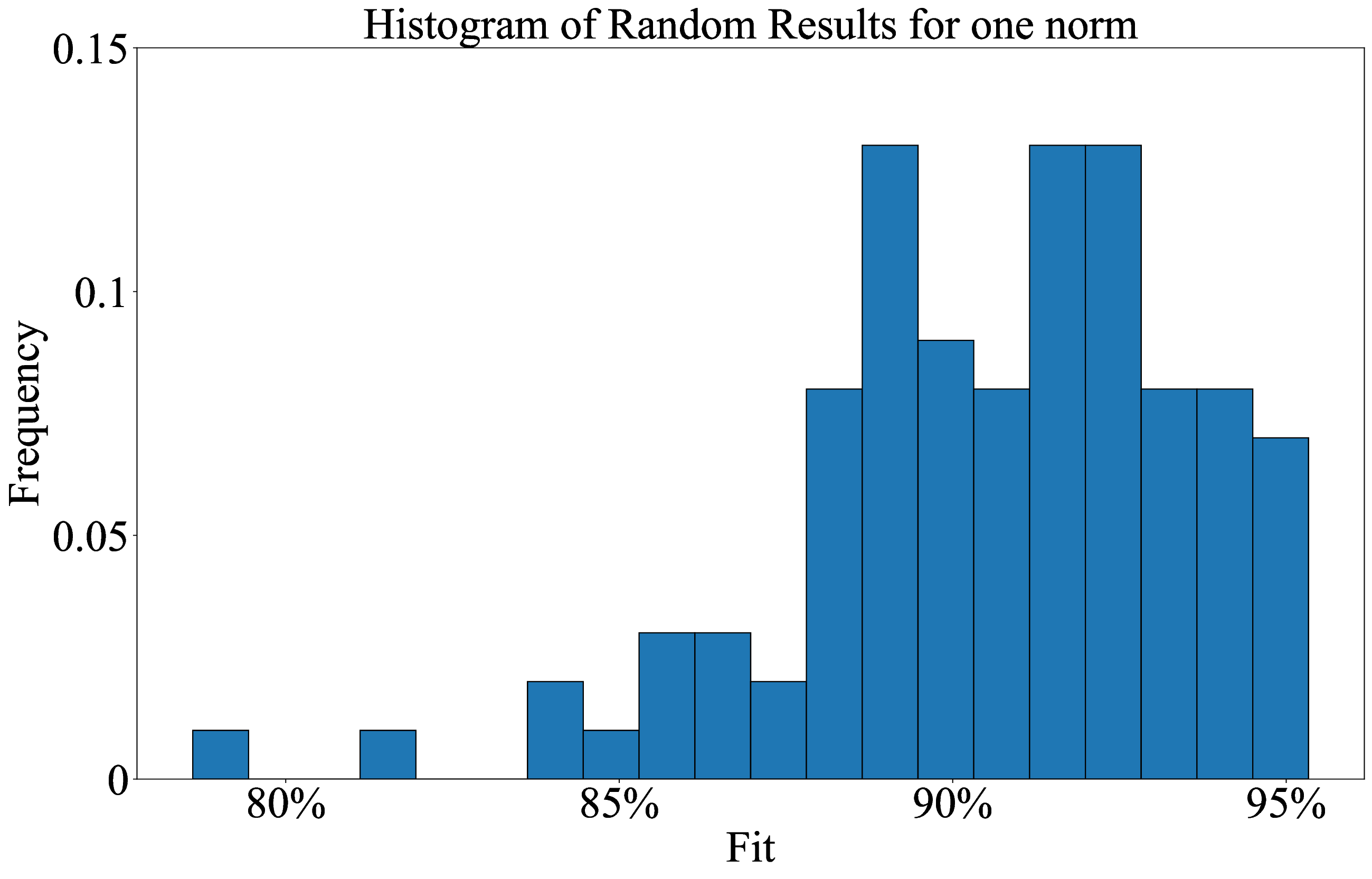} 
    \caption{The histogram of Fit for $100$ independent runs by Algorithm \ref{algo:com} with $\ell_1$-norm.} \label{fig:hist1}
\end{figure}






\subsection{Numerical Experiment with random switching instants}
In this section, experiments are conducted to test the identification of switching instants and its impact on the identification for the parameters of switched linear systems \eqref{s1}.
We consider a switched system \eqref{s1} with two subsystems:
\begin{equation}
    \begin{aligned}
        \theta_1 & =[-0.9,-0.2,0.16,0.2]^\top \\
        \theta_2 & = [-0.8,-0.1,0.26,0.15]^\top \\
    \end{aligned}
\end{equation}
where the number of the modes $S=2$, the mode $\lambda_k\in\{0,1\}$, the input $u_k\thicksim \mathcal{N}(0,1)$, noise $e_k\thicksim \mathcal{N}(0,\sigma^2)$ and the number of the time segments $M=15$. $14$ switching instants are randomly selected between $1$ and $1000$. The first $800$ samples are used to identify the switched linear system \eqref{s1} and the last $200$ samples are regarded as the test data.
In the first experiment, the output $y_k$ is noiseless, with $\sigma=0, SNR=\infty$, and in the second experiment, we set the variance of the noise $\sigma$ so that $SNR=10$. We also compute the number of the elements in the intersection of the kernel spaces corresponding to subsystem parameters $\theta_1, \theta_2$:
\begin{equation}
    \begin{aligned}
        & SNR=\infty:~|\{I(\theta_1)\cap I(\theta_2)\}|/N=51.7\%\\
        & SNR=10:~|\{I(\theta_1)\cap I(\theta_2)\}|/N=38.6\%.
    \end{aligned}
\end{equation}
The results of the identification for the switching instants based on both train data and test data are summarized in Table \ref{t1} and Table \ref{t2} while the results of the identified system parameters on test data are summarized in Table \ref{compare_fit2}. The prediction results are depicted in Fig. \ref{fig:infSNR} with $SNR=\infty$ and Fig. \ref{fig:10SNR} with $SNR=10$. The results of the above experiments indicate that our proposed identification algorithms for switching instants exhibit strong robustness and the robust identification of switching instants can assist in achieving higher accuracy in the identification of system parameters than the algorithm proposed in \cite{bako2011identification}.

\setlength{\tabcolsep}{5pt}
\begin{table}[h] 
    \centering 
    \caption{Identification for switching instants With $\text{SNR}=\infty$}\label{t1}
    \begin{tabular}{@{}lllcllllll@{}}
    \toprule
    Index &$m$ &2   & 3   & 4   & 5   & 6   & 7   & 8       \\ \midrule
    True &$s_m$ &34 & 57 & 237 & 295 & 451 & 605 & 636  \\
    Result &$\hat{s}_m$ &33 & 56 & 236 & 295 & 450 & 604 & 636   \\ \bottomrule
    Index &$m$ & 9 & 10  & 11  & 12  & 13  & 14 &15 &    \\ \midrule
    True &$s_m$ &715 &770 & 777 & 822 & 845 & 962 & 968 & \\
    Result &$\hat{s}_m$ &714 &769 & 777 & 821 & 844 & 961 & 968  & \\ \bottomrule
    \end{tabular}
\end{table}
\begin{table}[h] 
    \centering
    \caption{Identification for switching instants With $\text{SNR}=10$} \label{t2}
    \begin{tabular}{@{}lllcllllllllllllll@{}}
    \toprule
    Index &$m$ &2   & 3   & 4   & 5   & 6   & 7   & 8  \\ \midrule
    True &$s_m$ &39 & 59 & 239 & 269 & 439 & 579 & 599  \\
    Result &$\hat{s}_m$ &39 & 54 & 239 & 268 & 439 & 574 & 599  \\ \bottomrule
    Index &$m$ & 9 & 10  & 11  & 12  & 13  & 14 &15 & \\ \midrule
    True &$s_m$ &659 &729 & 749 & 779 & 809 & 939 & 969 \\
    Result &$\hat{s}_m$ &668 &728 & 751 & 778 & 806 & 935 & 969 \\ \bottomrule
    \end{tabular}
\end{table}
\begin{table}
    \centering
    \caption{Comparison of the proposed algorithms with random switching instants} \label{compare_fit2}
    \begin{tabular}{@{}lccc@{}}
        \toprule
        Algorithm & Sparse & Zero norm & One norm \\ \midrule
        Fit ($ SNR=\infty$) & 79.79.14\% & 96.03\% & 98.02\% \\
        Fit ($SNR = 10$)    & 31.22\% & 39.34\% & 79.01\% \\ \bottomrule
    \end{tabular}
\end{table}
\begin{figure} 
    \centering
    \includegraphics[width=\linewidth]{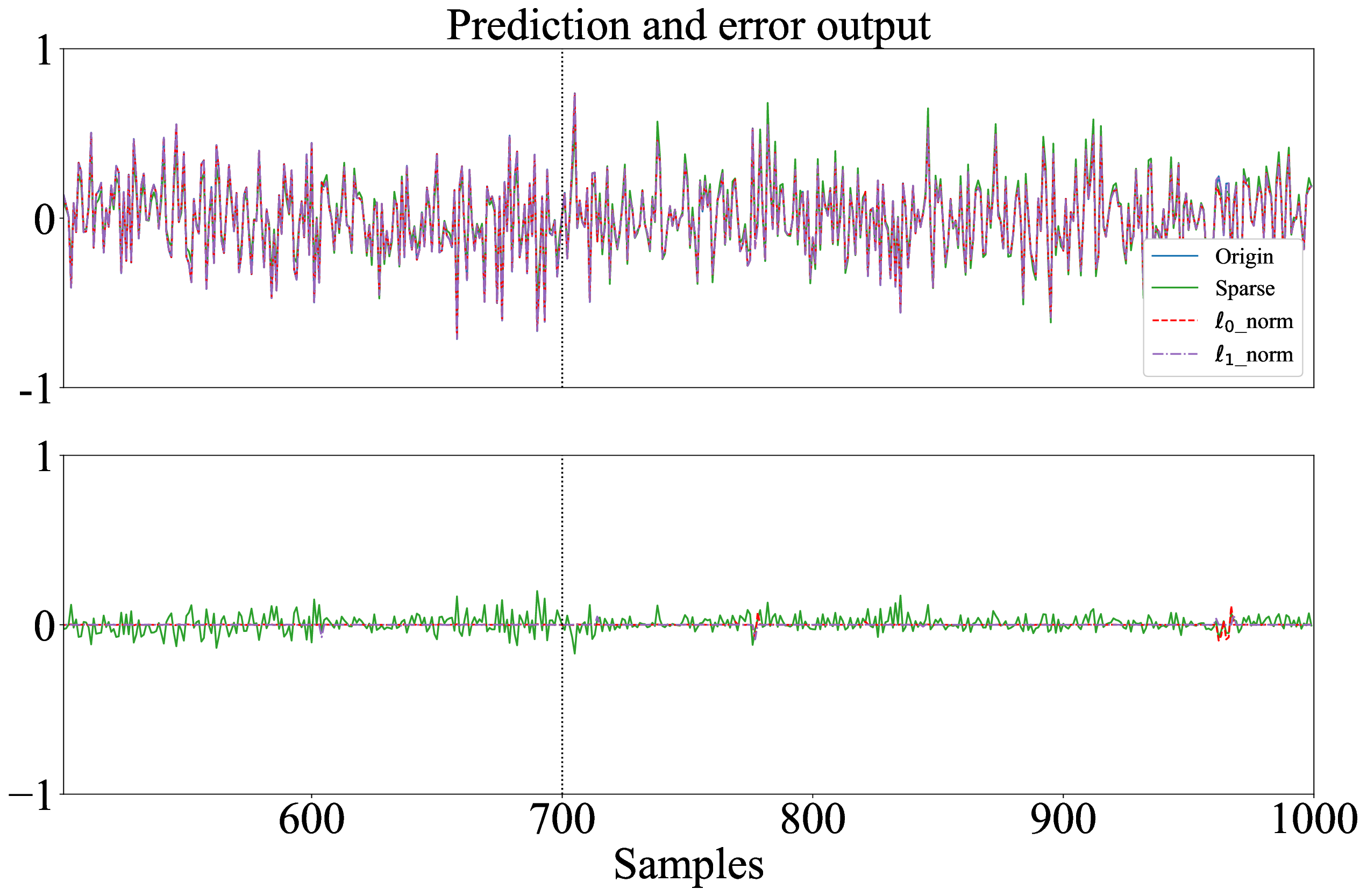} 
    \caption{500 identification and 200 test data with $SNR=\infty $.} \label{fig:infSNR}
\end{figure}
\begin{figure}
    \centering
    \includegraphics[width=\linewidth]{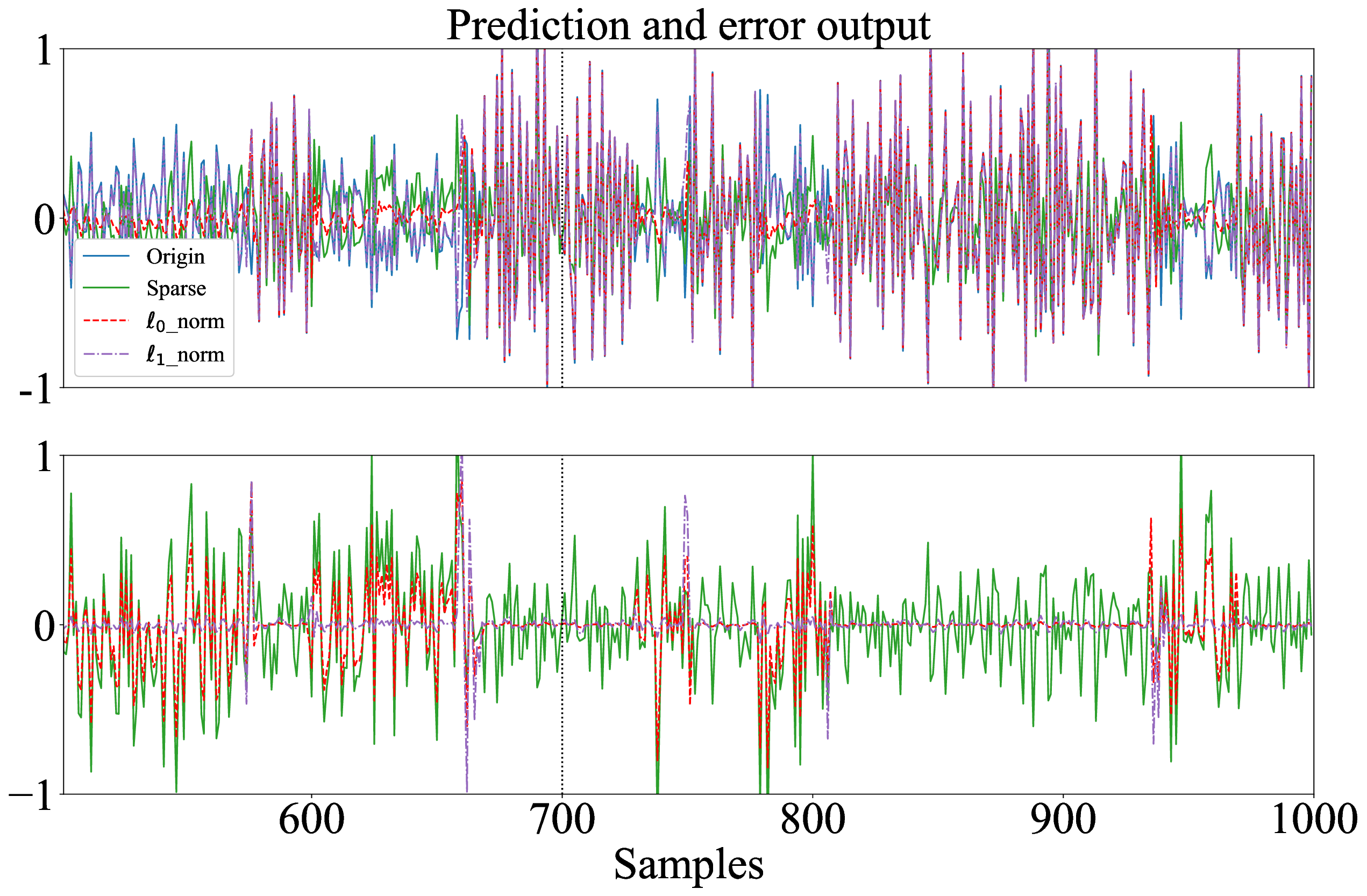} 
    \caption{500 identification and 200 test data with $SNR=10 $.}  \label{fig:10SNR}
\end{figure}

\subsection{Case study on high speed train modeling}

In this section, we adopt data collected from field trials on CRH3 high speed train \cite{sun2010optimization} in China to investigate the performance of our proposed Algorithm \ref{algo:com} as compared to the linear system by the least squares identification without data segmentation. 

High speed trains exhibit different modalities under various states and inputs, which can be modeled as a switched linear system. The output vector $y_k$ and the regressor vector $x_k$ in the switched linear system \eqref{s1} are defined as 
\begin{equation}
    \begin{aligned}
        y_k &= [v_k,s_k]^\top \\
        x_k &= [y_k^\top,\cdots,y_{k-5}^\top,\sigma_k,\cdots,\sigma_{k-5}]
    \end{aligned}
\end{equation}
where $v_k$, $s_k$ and $\sigma_k$ represents the velocity, position and the motor power output variable of the train. 
\begin{figure}
    \centering
    \includegraphics[width=\linewidth]{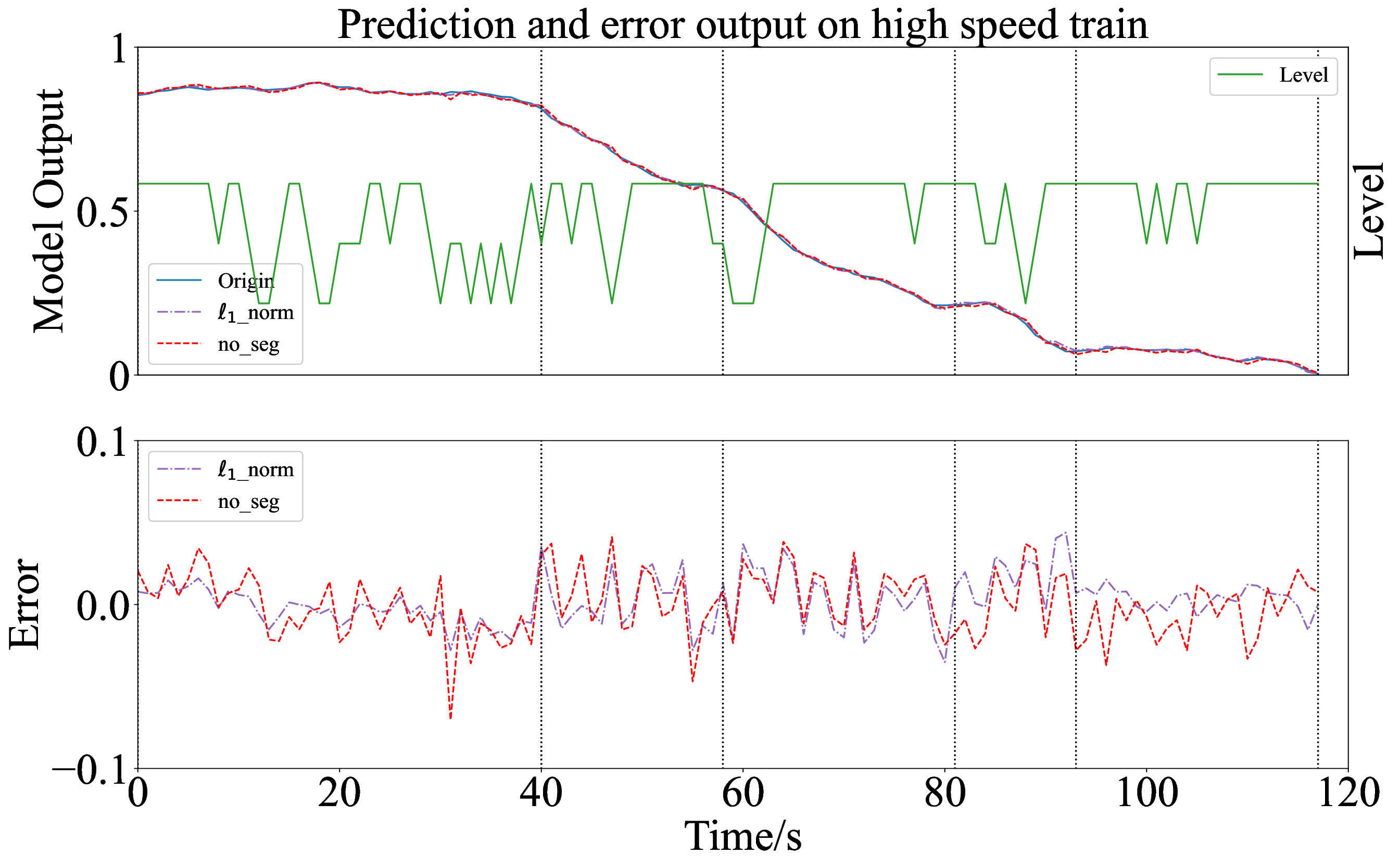} 
    \caption{High speed train modeling.}  \label{fig:hst}
\end{figure}
The results of model prediction and error output are depicted in Fig. \ref{fig:hst}.
It is evident that Algorithm \ref{algo:com} based on the switched linear system with data segmentation and sparsity inducing methods has much smaller fitting error than the traditional linear model. This indicates that the segmentation of time series data and the identification of switching rules can significantly improve the performance of high speed train modeling.

\section{Conclusion}

In this paper, we introduce a new two-stage identification algorithm for switched linear systems with constrained switching mechanisms. Initially, the switching instants are determined using the iterative dynamic programming method. Subsequently, system parameters are identified based on data segmentation between these switching instants using sparse inducing methods.
We formulate a specific combinatorial $\ell_0$-norm optimization problem, which can be relaxed into an $\ell_1$-norm optimization problem for ease of solution, to decouple the identification of discrete mode states and system parameters. Furthermore, we discuss the unbiasedness of our algorithm's identification in detail.
Additionally, we propose a new persistence of excitation criterion for switched linear systems based on the constrained switching mechanism. Experimental results demonstrate robustness in identifying switching instants, leading to improved prediction accuracy of the identified switched linear system.

\begin{ack}                               
    This work is supported in part by the National Natural Science Foundation of China under Grant 61933015 and Grant 61803163, in part by the Beijing National Research Center for Information Science and Technology (BNRist) Program under Grant BNR2019TD01009, and in part by the National Innovation Center of High Speed Train R\&D project ‘‘Modeling and comprehensive intelligent optimization for new high-efficiency urban rail transit system’’ (Grant No. CX/KJ-2020-0006).
\end{ack}

\bibliographystyle{unsrt} 
\bibliography{switchreview}

\appendix
\section{Improved OMP} \label{append_1}
The orthogonal matching pursuit (OMP) \cite{tropp2007signal,wang2012generalized} is proposed to address the normalized $\ell_0$-norm optimization problem \eqref{P_J} with $J(z)=\lVert z\rVert_0$. 
Although OMP is inherently greedy, the convergence and optimality of the solution to the general $\ell_0$-norm optimization problem against noise is guaranteed when condition \eqref{condition_mu} holds, as discussed in \cite{donoho2005stable}.
To compare with the proposed $\ell_1$-norm based algorithm, the OMP algorithm is improved to address the combinatorial $\ell_0$-norm optimization problem \eqref{prob_zero} with segmented data $\bar{D}=\{\bar{D}_1,\cdots,\bar{D}_M\}$. 
Notably, we focus on refining the "Sweep" step and the "Update Support" step, as the selection of basic elements now relies on segmented data rather than a single data sample. The improved OMP algorithm is outlined in the following Algorithm \ref{algo:zero}.

\begin{algorithm}[ht]
    \caption{Improved OMP} \label{algo:zero}
    \hspace*{0.02in} {\bf Require:} \\ 
    \hspace*{0.02in} $A_X\in \mathbb{R}^{(N-n)\times N},~b_X\in \mathbb{R}^{N-n}$, switching instants $\{ s_m \}_{m=1}^{M+1}$, segmented data $\bar{D}=\{\bar{D}_1,\cdots,\bar{D}_M\}$ and hyperparameters $\varepsilon_0$.
    \begin{algorithmic}[1]
    \State According to the switching instants $\{ s_m \}_{m=1}^{M+1}$ in column order, partition matrix $A_X$ into block matrices, $A_X = [A^{1}_X,\cdots,A^{M}_X]$, where $A^{m}_X \in \mathbb{R}^{(N-n)\times p_m}$.
    \State Initial settings: $j=0$
    \begin{itemize}
        \item Initial solution: $z^0=0$
        \item Initial residual: $r^0=b_X-A_X z^0$
        \item Initial solution support: $S^0=\{\emptyset\}$
    \end{itemize}
    \State While $\|r^j\|_2 \le \varepsilon_0$
    \begin{itemize}
        \item Sweep: Computing errors $\epsilon(m)$ with the optimal choice $q_m^*$
        \begin{equation}
            \begin{aligned}
                \epsilon(m) & = \min_{q_m\in \mathbb{R}^{p_m}} \|A_m q_m - r^{j-1} \|_2^2 \\
                q_m^* & = \arg\min_{q_m\in \mathbb{R}^{p_m}} \|A_m q_m - r^{j-1} \|_2^2 .
            \end{aligned}
        \end{equation}
        \item Update Support: Find a minimizer $m_0$ of $\epsilon(m)$ such that $\epsilon(m_0) \leq \epsilon(m), \forall m\notin S^{j-1}$ and update $S^j = S^{j-1}\bigcup\{s_{m_0},\cdots,s_{m_0+1}-1\}$
        \item Update provisional solution: Compute $z^j$ the minimizer of $\|A_Xz-b_X\|_2^2$ subject to Support$\{z\}=S^j$
        \item Update residual: Compute $r^j=b_X-A_Xz^j$
    \end{itemize}
    \State Concatenate the data belonging to the $i_0$th subsystem, $I(\theta_{i_0})=\{x_k:k\in S^j\}$, and re-estimate the parameters of the $i_0$th subsystem as follows:
    \begin{equation} \label{thetar}
        \begin{aligned}
            \theta_{i_0} = \arg\min_{\theta} \sum_{k\in S^j} \|z_k-x_k^\top\theta\|_2^2.
        \end{aligned}
    \end{equation}
    \State \Return $\theta_{i_0}, I(\theta_{i_0})$ 
    \end{algorithmic}
\end{algorithm}

\end{document}